\newif\ifFull
\newif\ifCameraReady
\Fulltrue

\newif\ifShowLineNum

\ifShowLineNum
\documentclass[a4paper,USenglish,cleveref, autoref, thm-restate]{socg-lipics-v2021} 
\usepackage{algorithm2e}
\else
\documentclass[a4paper,USenglish,cleveref, autoref, thm-restate,nolineno]{socg-lipics-v2021} 
\usepackage[linesnumbered]{algorithm2e}
\fi

\RequirePackage[T1]{fontenc}
\usepackage{subcaption}
\usepackage{xspace}
\usepackage{thm-restate}

\newcommand{\full}[2]{\ifFull #1\fi \ifCameraReady #2\fi\xspace}

\title{A Parallel Batch-Dynamic Data Structure for the Closest Pair Problem} 

\titlerunning{A Parallel Batch-Dynamic Data Structure for the Closest Pair Problem} 

\author{Yiqiu Wang}{Massachusetts Institute of Technology, United States}{yiqiuw@mit.edu}{}{}
\author{Shangdi Yu}{Massachusetts Institute of Technology, United States}{shangdiy@mit.edu}{}{}
\author{Yan Gu}{University of California, Riverside, United States}{ygu@cs.ucr.edu}{}{}
\author{Julian Shun}{Massachusetts Institute of Technology, United States}{jshun@mit.edu}{}{}

\authorrunning{Y. Wang, S. Yu, Y. Gu, and J. Shun} 

\Copyright{Yiqiu Wang, Shangdi Yu, Yan Gu, and Julian Shun} 



\category{} 

\full{
\relatedversion{}
}{
\relatedversion{}
\relatedversiondetails{Full Version}{https://arxiv.org/abs/2010.02379}
}

\supplement{}

\acknowledgements{We thank Yihan Sun for help on using the PAM library~\cite{sun2019parallel}. }

\funding{This research was supported by DOE
Early Career Award \#DESC0018947, NSF CAREER Award \#CCF-1845763, Google Faculty Research Award, DARPA SDH Award \#HR0011-18-3-0007, and Applications Driving
Architectures (ADA) Research Center, a JUMP Center co-sponsored by SRC
and DARPA.}

\EventEditors{Kevin Buchin and \'{E}ric Colin de Verdi\`{e}re}
\EventNoEds{2}
\EventLongTitle{37th International Symposium on Computational Geometry (SoCG 2021)}
\EventShortTitle{SoCG 2021}
\EventAcronym{SoCG}
\EventYear{2021}
\EventDate{June 7--11, 2021}
\EventLocation{Buffalo, NY, USA}
\EventLogo{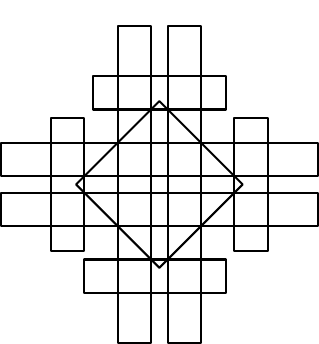}
\SeriesVolume{189}
\ArticleNo{36}

\newcommand{\defn}[1]{\emph{\textbf{#1}}}
\newcommand{\myparagraph}[1]{\medskip\noindent {\bf #1.}}

\newcommand{\whp}{\textit{whp}}

\newcommand{\kdt}{$k$\text{d-tree}\ }

\newcommand{\heapify}{\textsc{Heapify} }

\crefname{algorithm}{Algorithm}{Algorithms}

\newcommand{\down}{\mathit{down}}
\newcommand{\up}{\mathit{up}}
\newcommand{\hide}[1]{}

\newcommand{\s}[1]{S_{#1}}
\newcommand{\sDict}[1]{\s{#1}^\textsf{dict}}
\newcommand{\sprime}[1]{S'_{#1}}
\newcommand{\sprimeDict}[1]{S'_{#1}{}^\textsf{dict}}

\SetKwRepeat{Do}{do}{while}
\SetKwBlock{DoParallel}{do in parallel}{end}
\SetKwBlock{DoAsync}{do async}{end}
\SetKw{Spawn}{fork}
\SetKw{Sync}{join}

\begin{document}

\maketitle

\begin{abstract}
We propose a theoretically-efficient and practical parallel batch-dynamic data structure for the closest pair problem. Our solution is based on a serial dynamic closest pair data structure by Golin et al., and supports batches of insertions and deletions in parallel. For a data set of size $n$, our data structure supports a batch of insertions or deletions of size $m$ in $O(m(1+\log ((n+m)/m)))$ expected work and $O(\log (n+m)\log^*(n+m))$ depth with high probability, and takes linear space. 
The key techniques for achieving these bounds are a new work-efficient parallel batch-dynamic binary heap, and careful management of the computation across sets of points to minimize work and depth.

We provide an optimized multicore implementation of our data structure using dynamic hash tables, parallel heaps, and dynamic $k$-d trees. Our experiments on a variety of synthetic and real-world data sets show that it achieves a parallel speedup of up to 38.57x (15.10x on average) on 48 cores with hyper-threading.
In addition, we also implement and compare four parallel algorithms for static closest pair problem, for which we are not aware of any existing practical implementations. 
On 48 cores with hyper-threading, the static algorithms achieve up to 51.45x (29.42x on average) speedup, and Rabin's algorithm performs the best on average.
Comparing our dynamic algorithm to the fastest static algorithm, we find that it is advantageous to use the dynamic algorithm for batch sizes of up to 20\% of the data set.
As far as we know, our work is the first to experimentally evaluate parallel closest pair algorithms, in both the static and the dynamic settings.

\end{abstract}

\section{Introduction}

The closest pair problem is a fundamental computational geometry problem with applications in robot motion planning~\cite{Laskar2012,Arimoto1988}, computational biology~\cite{Rajasekaran2014}, collision detection, hierarchical clustering, and greedy matching~\cite{eppstein2000fast}. 
In many cases, the data involved can quickly change over time. In the case that a subset of the data gets updated, a dynamic algorithm can be much faster than a static algorithm that recomputes the result from scratch.

We consider a metric space $(S, d)$ where $S$ contains $n$ points in $\mathbb{R}^k$, and $d$ is the $L_t$-metric where $1\leq t\leq \infty$. 
The \defn{static closest pair} problem computes and returns
the \defn{closest pair distance}
$\delta(S) = \min\{d(p,q) \ |\ p,q\in S, p\neq q\}$, and the closest pair $(p,q)$.
The \defn{dynamic closest pair} problem computes the closest pair of $S$, and also maintains the closest pair upon insertions and deletions of points. A \defn{parallel batch-dynamic} data structure processes batches of insertions and deletions of points of size $m$ in parallel. 
In this paper, we propose a new parallel batch-dynamic data structure for closest pair on $(S, d)$.

There is a rich literature on sequential dynamic closest pair algorithms~\cite{overmars1981dynamization,supowit1990new,schwarz1991n,smid1990maintaining,smid1991maintaining,schwarz1994optimal,callahan1995algorithms,kapoor1996new,bespamyatnikh1998optimal,golin1998randomized,agarwal2008kinetic}. More details on this related work is provided in~\cite{smid2000handbook} and \full{Appendix~\ref{sec:related}}{the full version of this paper}.
However, as far as we know, none of the existing dynamic algorithms have been implemented and none of them are parallel.
The main contribution of our paper is the design of a theoretically-efficient and practical parallel batch-dynamic data structure for dynamic closest pair, along with a comprehensive experimental study showing that it performs well in practice. 
Our solution is inspired by the sequential solution of Golin et al.~\cite{golin1998randomized}, which takes $O(n)$ space to maintain $O(n)$ points and supports $O(\log n)$ time updates, and is the fastest existing sequential algorithm for the $L_t$-metric.
Our parallel solution takes a batch update of size $m$ and maintains the closest pair in $O(m(1+\log ((n+m)/m)))$ expected work and $O(\log (n+m)\log^*(n+m))$ depth (parallel time) with high probability (\whp).\footnote{Holds with
probability at least $1-1/n^c$ for an input of size $n$ and some constant $c>0$.} Compared to the sequential algorithm of Golin et al., our algorithm is work-efficient (i.e., matches the work of the sequential algorithm) for single updates, and has better depth for multiple updates since we process a batch of updates in parallel.
Our data structure is based on efficiently maintaining a sparse partition of the points (a data structure used by Golin et al.~\cite{golin1998randomized}) in parallel. This requires carefully organizing the computation to minimize the work and depth, as well as using a new parallel batch-dynamic binary heap that we design. 
This is the first parallel batch-dynamic binary heap in the literature, and may be of independent interest. 

We implement our data structure with optimizations to improve performance. In particular, we combine the multiple heaps needed in our theoretically-efficient algorithm into a single heap, which reduces overheads. We also implement a parallel batch-dynamic \kdt to speed up neighborhood queries. 
We evaluate our algorithm on both real-world and synthetic data sets.
On 48 cores with two-way hyper-threading, we achieve self-relative parallel speedups of up to 38.57x across various batch sizes.
Our algorithm achieves throughputs of up to
$1.35\times 10^7$ and $1.06\times 10^7$ updates per second for insertions and deletions, respectively.

In addition, we implement and evaluate four parallel algorithms for the static closest pair problem. There has been significant work on sequential~\cite{shamos1975closest,rabin1976probabilistic,bentley1976divide,fortune1979note,bentley1980multidimensional,hinrichs1988plane,golin1992simple,khuller1995simple,dietzfelbinger1997reliable,chan1999geometric,banyassady2017simple} and parallel~\cite{atallah1986efficient,mackenzie1998ultrafast,lenhof1995sequential,blelloch2010parallel,BGSS16} static algorithms for the closest pair (more details can be found in~\cite{smid2000handbook} and 
\full{Appendix~\ref{sec:related}}{the full version of the paper}).
As far as we know, none
of the existing parallel algorithms have been evaluated and compared empirically.
We implement a divide-and-conquer algorithm~\cite{blelloch2010parallel}, a variant of Rabin's randomized algorithm~\cite{rabin1976probabilistic},
our parallelization of the sequential sieve algorithm~\cite{khuller1995simple}, and a randomized incremental algorithm~\cite{BGSS16}.
On 48 cores with two-way hyper-threading, our algorithms achieve self-relative parallel speedups of up to 51.45x.\footnote{With hyper-threading, the parallel speedup can be more than the total core count.} 
Our evaluation of the static algorithms shows that Rabin's algorithm is on average 7.63x faster than the rest of the static algorithms.
Finally, we compare our parallel batch-dynamic algorithm with the static algorithms and find that it can be advantageous to use the batch-dynamic algorithm for batches containing up to 20\% of the data set.
Our source code is publicly available at {\url{https://github.com/wangyiqiu/closest-pair}}.


\section{Review of the Sparse Partition Data Structure}\label{sec:golin}

We give an overview of the sequential dynamic closest pair data structure proposed by Golin et al.~\cite{golin1998randomized}, which is based on the serial static closest pair algorithm by Khuller and Matias~\cite{khuller1995simple}. More details are presented in \full{Appendix~\ref{appendix:golin}.}{the full version of the paper.}
Our new parallel algorithm also uses this data structure, which is referred to as the \emph{sparse partition} of an input set.

\subsection{Sparse Partition}\label{sec:golin:sp}

For a set $S$ with $n$ points, a \defn{sparse partition}~\cite{golin1998randomized} is defined as a sequence of 5-tuples $(S_i,S_i',p_i,q_i,d_i)$ with size $L$ ($1\leq i\leq L$), such that
\textbf{(1)} $S_1=S$;
\textbf{(2)} $S_i' \subseteq S_i \subseteq S$;
\textbf{(3)} If $|S_i|>1$, $p_i$ is drawn uniformly at random from $S_i$, and we compute the distance $d_i=d(p_i,q_i)$, to $p_i$'s closest point $q_i$ in $S_i$;
\textbf{(4)} For all $x\in S_i$:
\textbf{a)} if the closest point of $x$ in $S_i$ (denoted as $d(x,S_i)$) is larger than $d_i/3$, then $x\in S_i'$; \textbf{b)} if $d(x,S_i)\leq d_i/6k$ then $x\notin S_i'$;
and \textbf{c)} if $x\in S_{i+1}$, there is a point $y\in S_i$ such that $d(x,y)\leq d_i/3$ and $y\in S_{i+1}$; 
\textbf{(5)} $S_{i+1}=S_i\setminus S_i'$.

The sparse partition is constructed using these rules until $S_{L+1}=\emptyset$.
It contains $O(\log n)$ levels in expectation, as $|S_i|$ decreases geometrically.  The expected sum of all $|S_i|$ is linear~\cite{golin1998randomized}.
We call $p_i$ the \defn{pivot} for partition $i$.
At a high level, $S_i'$ contains points that are far enough from each other, and the threshold $d_i$ that defines whether points are "far enough" decreases for increasing $i$. In particular, for any $1\leq i<L$, $d_{i+1}\leq d_i/3$ as shown in Golin et al.~\cite{golin1998randomized}.
Hence, the closest pair will likely show up in deeper levels that do not contain many points.
Based on the definition, each $S_i'$ is non-empty, and
$\{S_1', \ldots, S_L'\}$ is a partition of $S$. 

\subsection{A Grid-Based Implementation of Sparse Partition}
We now describe Golin et al.'s implementation of the sparse partition.
There are $L$ levels of the sparse partition, and we refer to each as \defn{level $i$} for $1\leq i\leq L$. 
We maintain each level using a grid data structure, similar to many closest pair algorithms (e.g.,~\cite{rabin1976probabilistic,golin1992simple, khuller1995simple, golin1998randomized}).

To represent $S_i$, we place the points into a grid $G_i$ with equally-sized axis-aligned grid boxes with side length $d_i/6k$, where  $k$ is the dimension, and $d_i$ is the closest pair distance of the randomly chosen pivot $p_i$. This can be done using hashing.
Denote the \defn{neighborhood} of a point $p$ in $G_i$ relative to $S$ by $N_i(p,S)$, which refers to the set of points in $S\setminus \{p\}$ contained in the collection of $3^k-1$ boxes bordering the box containing $p$, as well as $p$'s box. We say that point $p$ is \defn{sparse} in $G_i$ relative to $S$ if $N_i(p,S)=\emptyset$.
We use this notion of sparsity to compute $S_i'= \{p\in S_i : p\text{ is sparse in }G_i\text{ relative to }S_i\}$, which satisfies definition (4) of the sparse partition. The points in $S_i'$ are stored in a separate grid. 

\begin{figure}[t]
  \includegraphics[trim=20 10 0 0,clip,width=1.02\textwidth]{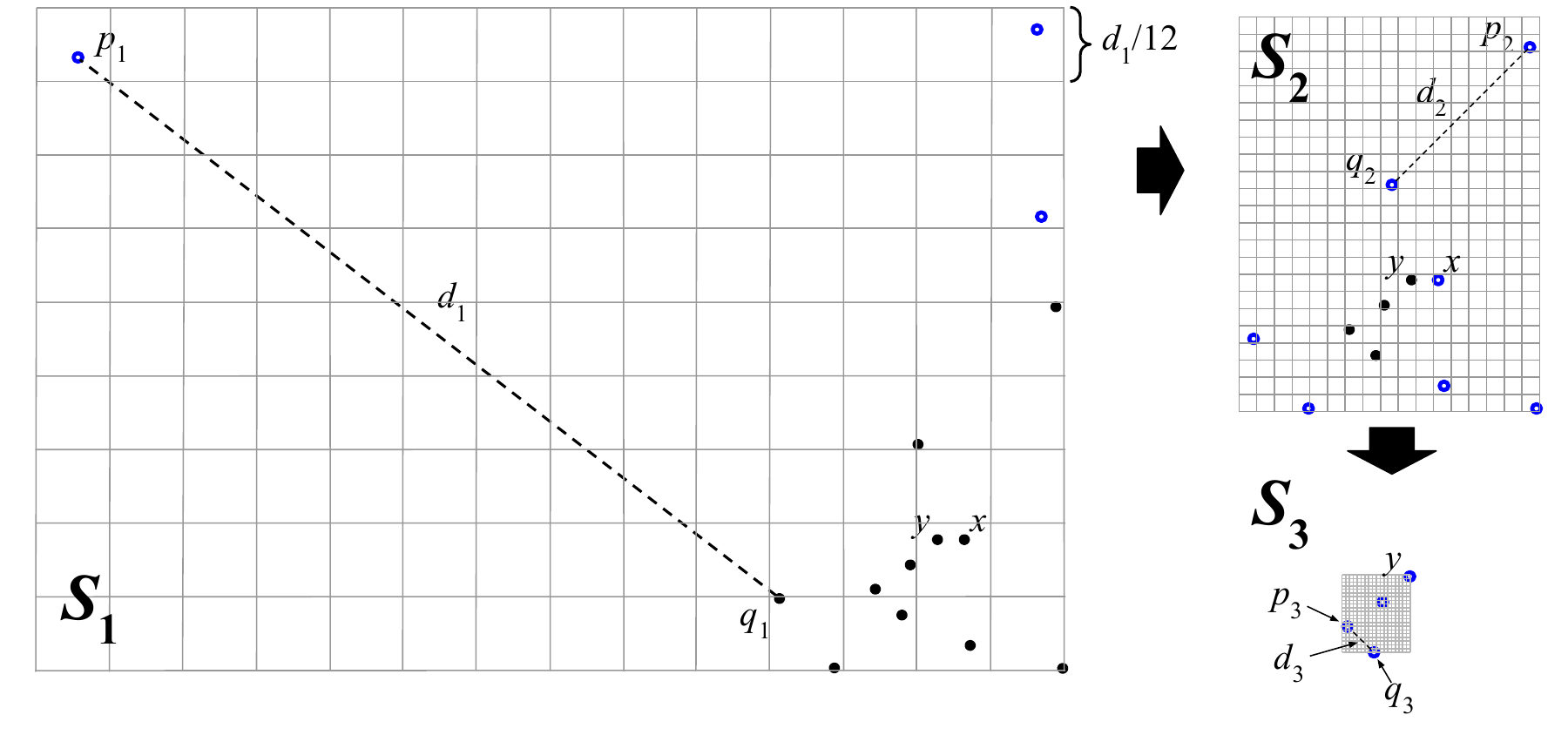}
  \caption{This figure contains an example of 14 points in $\mathbb{R}^2$, for which a grid-based sparse partition $(S_i,S_i',p_i,q_i,d_i)$ for $1\leq i\leq 3$ is constructed.
  On each level, we use a dotted line to indicate $d_i$, the Euclidean distance between the pivot $p_i$ and its closest neighbor $q_i$, and we set the grid size to be $d_i/6k=d_i/12$.
  We denote non-sparse points as solid black circles and sparse points as hollow blue circles. The $S_i'$ sets are represented implicitly by the set of hollow blue circles in each $S_i$.
  We denote the true closest pair by letters $x$ and $y$.
  }
  \label{fig:example}
\vspace{1em}

\end{figure}

An example of the grid structure in two dimensions is shown in Figure~\ref{fig:example}. We illustrate the grid $G_i$ for the $S_i$ of each level, as well as the pivot $p_i$ and its closest neighbor $q_i$. The grid size is set to $d_i/6k=d(p_i,q_i)/12$. The sparse points, represented by the hollow blue circles, have empty neighborhoods, and do not have another point within a distance of $d_i/3$.
The solid black circles, representing the non-sparse points, are copied to the grid $G_{i+1}$ for $S_{i+1}$.
In $S_3$, all points are sparse.

To construct a sparse partition, the sequential algorithm proceeds in rounds. On round $i$, the $i$'th level is constructed. 
We start with $i=1$ where $S_1 = S$, and iteratively determine the side length of grid $G_i$ based on a random pivot, and place $S_i$ into $G_i$. 
Then, we compute $S_i'$ based on the definition of sparsity above, and set $S_{i+1}=S_i\setminus S_i'$. 
The algorithm proceeds until $S_i=S_i'$ (i.e., $S_{i+1}=\emptyset$).
The expected work for construction is $O(n)$ since $|S_i|$ decreases geometrically~\cite{golin1998randomized}.
A single insertion of point $q$ starts from $S_1$, and proceeds level by level. When $q$ is non-sparse in $S_i$, it will be added to $S_{i+1}$, and can promote points from $S_i'$ to $S_{i+1}'$ if $q$ falls within their neighborhood. The insertion of $q$ will stop if it becomes sparse at some level, at which point the insertion algorithm finishes. A deletion works in the opposite direction, starting from the last level where the deleted point exists, and working its way back to level $1$. Each insertion or deletion takes $O(\log n)$ expected work.

\subsection{Obtaining the Closest Pair}

As observed by both Khuller and Matias~\cite{khuller1995simple} and Golin et al.~\cite{golin1998randomized}, although the grid data structure rejects far pairs, and becomes more fine-grained with a larger $i$, the grid at the last ($L$'th) level does not necessarily contain the closest pair. For example, as illustrated in Figure~\ref{fig:example}, $S_3$ for the last level does not contain the closest pair $(x,y)$, as $x$ is sparse on level 2 and not included in $S_3$. Therefore, we need to check more than just the last level.

The \defn{restricted distance} $d^*_i(p)$~\cite{golin1998randomized} is the closest pair distance from point $p$ to any point in $\bigcup_{0\leq j \leq k} S_{i-j}'$, and defined as $d^*_i(p):= \min\{d_i, d(p, S_{i-k}'\cup S_{i-k+1}' \cup \ldots \cup S_i')\}$, where $p\in S_i'$.
Golin et al.\ show that $\delta(S) = \min_{L-k\leq i\leq L} \min_{p\in S_i'} d_i^*(p)$, meaning that the closest pair can be found by taking the minimum among the restricted distance pairs for all points in the last $k+1$ levels of $S_i'$. 
The sequential algorithm~\cite{golin1998randomized} computes the restricted distance for each point in $S_i'$, and stores it in a min-heap $H_i$ for level $i$. 
To obtain the closest pair, we read the minimum of $H_i$ for $L-k\leq i \leq L$, and then take the overall minimum. This takes $O(1)$ work.
In \full{Appendix~\ref{appendix:golin},}{the full version of the paper,} we provide more background on the restricted distance.

\section{Computational Model}

We use the classic
\defn{work-depth model} for analyzing parallel
algorithms~\cite{CLRS,JaJa92}.
The \defn{work} $W$ of an algorithm is the number of instructions in
the computation, and the \defn{depth} $D$ is the length of the longest sequential
dependence.  Using Brent's scheduling theorem~\cite{Brent1974},  we can
execute a parallel computation in $W/p+D$ running time using $p$ processors.
A parallel algorithm is \defn{work-efficient} if its work asymptotically matches the work of the best sequential algorithm for the same problem.
We assume that arbitrary concurrent writes are supported in $O(1)$ work and depth.
Our pseudocode uses the \textbf{fork} and \textbf{join} keywords for fork-join parallelism~\cite{CLRS}. A \textbf{fork} creates a task that can be executed in parallel with the current task, and a \textbf{join} waits for all tasks forked by the current task to finish.
In \full{Appendix~\ref{appendix:prelims},}{the full version of the paper,} we present details on the parallel primitives that we use.

For our batch-dynamic data structure, we assume that the updates are independent of the random choices made in our data structure.

\section{Parallel Batch-Dynamic Data Structure}\label{sec:par-alg}

In this section, we introduce our parallel batch-dynamic algorithm, including the construction of the sparse partition (defined in Section~\ref{sec:golin}) and how to handle batch updates. 

\begin{algorithm}[t]
  \caption{Construction}\label{alg:build-short}
  \SetAlgoLined
  \SetKwInOut{Input}{Input}
  \SetKwInOut{Output}{Output}
  \Input{Point set $S$.}
  \Output{A sparse partition and its associated heaps.} 

  \SetKwFunction{algo}{algo}\SetKwFunction{proc}{proc}
  \SetKwProg{myalg}{Algorithm}{}{}
  \SetKwProg{myproc}{Procedure}{}{}

  \myalg{\upshape\textsc{Main}()}{
    \textsc{Build}($S$, $1$);\label{codeline:build-short:init} \ \tcc{Initially, $\s{i}:=S$.}
  }{}
  \myproc{\upshape\textsc{Build}($\s{i}$, $i$)}{
    Choose a random point $p_i\in \s{i}$. Calculate $d_i:=d(p_i, \s{i})$, set the grid side length to $d_i/6k$, and store $p_i$'s nearest neighbor as $q_i$.\label{codeline:build-short:pivot}

    Create a parallel dictionary $\sDict{i}$ to store points in $\s{i}$ keyed by box ID.
    In parallel, compute the box ID of each point in $\s{i}$ based on the grid size, and store the point in the box keyed by the box ID in $\sDict{i}$. 
    \label{codeline:build-short:si}

    Create a parallel dictionary $\sprimeDict{i}$ to store points in $\sprime{i}$ keyed by box ID.
    In parallel, determine if each point $x$ in $\s{i}$ is sparse by checking $N_i(q,\s{i})$ (using $\sDict{i}$). Store the sparse points in $\sprime{i}$ and $\sprimeDict{i}$, and the remaining points in a new point set represented by an array $\s{i+1}$. 
    \label{codeline:build-short:siprime}

    In parallel for each point $x\in \sprime{i}$, compute $d_i^*(x)$ by checking its neighborhoods (using $\sprimeDict{j}$) in $\sprime{j}$ where $i-k\leq j \leq i$ .\label{codeline:build-short:sparsity}

    \Spawn Create a heap for $ \{d_i^*(x): x \in \sprime{i}\}$.\label{codeline:build-short:spawn}

    \Indp
    \textsc{Build}($\s{i+1}$, $i+1$) if $\s{i+1}$ is not empty.\label{codeline:build-short:recur}

    \Indm \Sync \label{codeline:build-short:sync}
  }
\end{algorithm}

As shown in~\cref{alg:build-short},
we start with an initial point set $S$, on which we construct a grid structure recursively level by level, until all points become sparse.
Starting with $\s{i}=S$, the algorithm works on point set $\s{i}$ for level $i$. We first pick a pivot point $p_i\in \s{i}$ and obtain its closest pair by computing distances to all other points in parallel, followed by computing the minimum to determine the side length of the grid boxes, which we use a parallel dictionary~\cite{Gil91a} to store (Lines~\ref{codeline:build-short:pivot}--\ref{codeline:build-short:si}). We check the sparsity of each point $x$ in parallel by looking up neighboring boxes using the dictionary, and store the sparse points $\sprime{i}$ in a new parallel dictionary and the remaining points in a new point set array $\s{i+1}$ (Line~\ref{codeline:build-short:siprime}).
Then, we compute the restricted distances of all points in $\sprime{i}$ in parallel (Line~\ref{codeline:build-short:sparsity}), and spawn a thread to asynchronously construct the heap $H_i$ to store the restricted distances (Line~\ref{codeline:build-short:spawn}).
We recursively call the construction procedure on $\s{i+1}$ to construct the next level until all points in a level are sparse (Line~\ref{codeline:build-short:recur}).
In \full{Appendix~\ref{appendix:construction},}{the full version of the paper,} we include more details about the algorithm, prove a high probability bound on the number of levels, and explain how to achieve optimal linear work and space. We summarize our bounds in the following theorem.

\begin{theorem}\label{theorem:build}
We can construct a data structure that maintains the closest pair containing $n$ points in $O(n)$ expected work, $O(\log n \log^* n)$ depth \whp, and $O(n)$ expected space.
\end{theorem}

\begin{figure}[]
  \includegraphics[trim=0 0 0 0,clip,width=\textwidth]{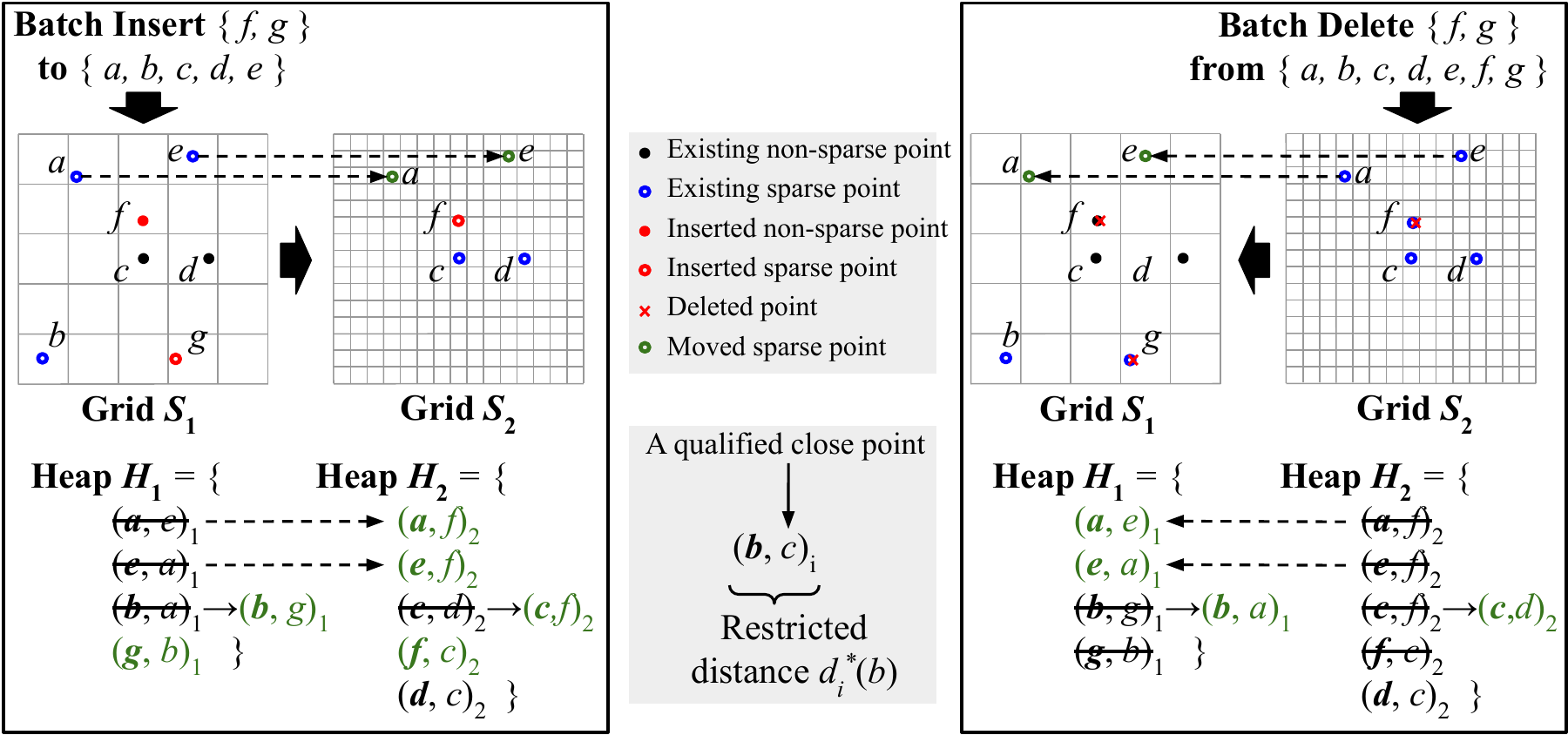}
  \caption{The figure illustrates the interaction between our parallel batch-dynamic insertion (left) and deletion (right) algorithms with the data structure. For ease of illustration, we do not show all the points in the data set. We show the data structure with two levels, and explicitly show $S_i$ and $H_i$ for each level. The grid structure in the upper half of the figures determines the sparsity of points. We represent different types of points as defined in the middle legend. In the lower half of the figures, we show the heaps with the restricted distances that they store. We show the pair defining the restricted distance of a sparse point $x$ on level $i$ as $(\textit{\textbf{x}},y)$ if another point $y$ is the closest sparse point to $x$ in levels $i-j$ where $0\leq j\leq 2$.
  For both insertion and deletion, we annotate the direction of the update between grids using large bold arrows (i.e., insertion starts with $S_1$ and deletion starts with $S_2$). We indicate the movement of points and heap entries using dotted arrows.
  }
  \label{fig:overview}
\end{figure}
\begin{algorithm}[]
\SetAlgoLined
  \SetKwInOut{Input}{Input}\SetKwInOut{Output}{Output}
  \Input{$(S_i, S_i', p_i, q_i, d_i)$ and $H_i$ for $1\leq i \leq L$; a batch $Q$ to be inserted.}

 \SetKwFunction{algo}{algo}\SetKwFunction{proc}{proc}
 \SetKwProg{myalg}{Algorithm}{}{}
 \SetKwProg{myproc}{Procedure}{}{}

 \myalg{\upshape\textsc{Main}()}{
   \textsc{Insert}($Q$, $\emptyset$, $1$)\;\label{line:grid-insert-short:call}
 }{}
   
 \myproc{\upshape\textsc{Insert}($Q_{i}$, $\down_{i}$, $i$)}{\label{line:grid-insert-short:insert-proc}
   $(Q_{i+1}$, $\down_{i+1}) :=$ \textsc{GridInsert}$(Q_{i}$, $\down_{i}$, $i)$\;\label{line:grid-insert-short:grid}
   \textsc{HeapUpdate}$(i)$\;\label{line:grid-insert-short:heap}
   \lIf{$(Q_{i+1}\cup \down_{i+1}) \neq \emptyset$}{
     \textsc{Insert}($Q_{i+1}$, $\down_{i+1}$, $i+1$)
   }\label{line:grid-insert-short:recur-insert}
 }
 
 \myproc{\upshape\textsc{GridInsert}$(Q_{i}$, $\down_{i}$, $i)$}{
   Determine if $p_i$, $q_i$, and $d_i$ should change when inserting $Q_i$ and $down_i$, which happens with probability $(|Q_i|+|down_i|)/(|Q_i|+|down_i|+|S_i|)$, or if a new point is closer to $p_i$ than the previously closest point $q_i$. \label{line:grid-insert-short:pivot}

    If $p_i$, $q_i$, or $d_i$ change on Line~\ref{line:grid-insert-short:pivot}, or if $i>L$, call \textsc{Build}$(Q_{i}\cup \down_{i}\cup S_i, i)$ to rebuild subsequent levels, and terminate the batch insertion.\label{line:grid-insert-short:rebuild}

   Insert each point in $\down_{i}$ and $Q_{i}$ into the dictionary of $S_i$ in parallel.\label{line:grid-insert-short:insert}

   For each point $x$ in $Q_{i}$ in parallel, check if it is sparse in $S_i$. If so, insert $x$ into the dictionary of $S_i'$, and otherwise, insert $x$ into $Q_{i+1}$.\label{line:grid-insert-short:check1}

   For each point $x$ in $\down_{i}$ in parallel, check if it is sparse in $S_i$. If so, insert $x$ into the dictionary of $S_i'$, and otherwise, insert $x$ into $\down_{i+1}$.\label{line:grid-insert-short:check2}

   In parallel, for each point $x$ in $Q_{i}$, and for each point $r$ in the neighborhood $N(x, S_i')$, delete $r$ from $S_i'$, and insert $r$ into $\down_{i+1}$.\label{line:grid-insert-short:check3}

   \Return $(Q_{i+1},\down_{i+1})$\;
 }

 \caption{Batch Insert}\label{alg:grid-insert-short}
\end{algorithm}

Next, we present our parallel  algorithm that processes a batch $Q$ of $m$ insertions or deletions. For $m\ge n$, we can simply rebuild the data structure on all of the points using \cref{theorem:build} to obtain the desired bounds. We now describe the case for $m < n$. 
For batch updates, there are two main tasks: updating the grid and updating the heap. We first describe updating the grid.
We let $Q_i$ be the subset of points in $Q$ that are inserted at level $i$, and $\down_{i}$ be the set of points that move from level $i-1$ to level $i$ due to the insertion of $Q_i$.
We start with a simple example of an insertion in \cref{fig:overview}~(left), which originally contains five points $\{a,b,c,d,e\}$. For simplicity, we assume that the pivot remains unchanged and also omits $S_i'$.
$Q_1=\{f,g\}$ is the set of points inserted into the grid at level $1$. We first update $S_1$ to include $f$ and $g$, and then update $S_2$ to include $Q_2=\{f\}$ but not $g$, since $g$ is already sparse in $S_1$.
In the example, the insertion of $Q_1$ triggers further point movements of $\{a,e\}$ from level $1$ to level $2$, as the sparse points $a$ and $e$ in $S_1$ become non-sparse due to the insertion of $f$.

We explain the insertion algorithm in detail, and defer most details of the deletion algorithm to \full{Appendix~\ref{appendix:deletion}}{the full version of the paper}.
As shown in Algorithm~\ref{alg:grid-insert-short}, the update proceeds recursively level by level (Lines~\ref{line:grid-insert-short:insert-proc}--\ref{line:grid-insert-short:recur-insert}).
Each call to the procedure \textsc{Insert}$(Q_{i}, \down_{i}, i)$ updates $(S_i,S_i',p_i,q_i,d_i)$ and $H_i$.
Initially, $Q_1=Q$ and $\down_1$ is empty, as shown on Line~\ref{line:grid-insert-short:call}.
For each level $i$, we update the pivot and rebuild the level with probability $(|Q_i|+|down_i|)/(|Q_i|+|down_i|+|S_i|)$ to ensure that the pivot is still  selected uniformly at random among the points in $S_i$, and we also update the pivot if a new point is closer to $p_i$ than the previous closest point $q_i$ (Lines~\ref{line:grid-insert-short:pivot}--\ref{line:grid-insert-short:rebuild}). 
Otherwise, we insert the points in both $\down_{i}$ and $Q_{i}$ into the dictionary representing $S_i$. We then check if the points that we inserted are sparse, and insert the sparse ones into the dictionary representing $S_i'$. The points that are not sparse will be added to sets $\down_{i+1}$ and $Q_{i+1}$ and passed on to the next level (Lines~\ref{line:grid-insert-short:insert}--\ref{line:grid-insert-short:check2}).
We then determine additional elements of $\down_{i+1}$ by including the points in the neighborhood of $Q_{i}$ 
in $S_i'$ (Line~\ref{line:grid-insert-short:check3}). 
In general, $\down_{i+1}$ is computed by $\down_{i+1} = \{x \ |\ x\in N_{i}(q,S_{i}'\cup \down_{i})$ for some $q\in Q_{i} \}$.
If $Q_{i+1}$ and $\down_{i+1}$ are empty, nothing further needs to be done for subsequent levels, and the tuples $(S_l, S_l', p_l, q_l, d_l)$ for $i< l\leq L$ remain unchanged. 

We now argue that the algorithm is correct.
Consider a round $i$ that inserts a non-empty $Q_{i}\cup \down_{i}$. After the insertion, the pivot is still chosen uniformly at random, since on Line~\ref{line:grid-insert-short:pivot}, we choose $p_i$ such that each point in $S_i \cup Q_{i} \cup \down_{i}$ has the same probability of being chosen.
All sparse points in $Q_{i}$ and $\down_{i}$ inserted into $S_i$ are included in $S_i'$ (Lines~\ref{line:grid-insert-short:check1}--\ref{line:grid-insert-short:check2}). Line~\ref{line:grid-insert-short:check3} additionally ensures that all points that were originally sparse in $S_i'$, but are no longer sparse after the insertion, are removed from $S_i'$.
Given that the non-sparse points in the original $S_i$ will not become sparse due to the batch insertion, $S_i'$ must contain exactly all of the sparse points of the updated $S_i$.

In our algorithm, each point can be moved across multiple levels as a result of a batch insertion.
\full{In Appendix~\ref{appendix:insertion},}{In the full version of the paper,}
we prove the following lemmas, which are key to maintaining work-efficiency.

\begin{restatable}{lemma}{packing}\label{lemma:packing}
  $|\bigcup_{1\leq i\leq L} \down_i| \leq m\cdot 3^k = O(m)$
\end{restatable}

\begin{restatable}{lemma}{sumq}\label{lemma:sum-q}
  $\sum_{1\leq i\leq L}E[|Q_i|] = O(m)$
\end{restatable}

Deletions work similarly in the reverse direction as shown in Figure~\ref{fig:overview}~(right), and we provide more details in \full{\cref{appendix:deletion}.}{the full version of the paper.} We obtain Lemmas~\ref{lemma:grid-update} and~\ref{lemma:grid-update-delete}.

\begin{lemma}\label{lemma:grid-update} 
We can maintain a sparse partition for a batch of $m$ insertions in $O(m)$ amortized work in expectation and $O(\log (n+m) \log^*(n+m))$ depth \whp. 
\end{lemma} 

\begin{proof}

The expected cost of rebuilding on Line~\ref{line:grid-insert-short:rebuild} summed across all rounds is proportional to the batch size. First, we re-select the pivot and rebuild with probability $(|Q_{i}|+|\down_{i}|)/(|S_i|+|Q_{i}|+|\down_{i}|)$.
When the pivot $p_i$ is unchanged, it may update its closest point to $q_i^*$ from $Q_{i}\cup \down_{i}$.
It is easy to show that $q_i^*$ can be the nearest neighbor of at most $3^k-1$ points in $S_i$. Hence, considering all candidates $Q_{i}\cup \down_{i}$, it follows that they can be the nearest neighbors to $O(3^k\cdot (|Q_{i}|+|\down_{i}|))$ points in $S_i$. Therefore, the pivot distance changes with probability at most $3^k\cdot (|Q_{i}|+|\down_{i}|)/|S_i|$, in which case we rebuild the sparse partition. The expected work of rebuilding at level $i$ is
$O((|S_i|+|Q_{i}|+|\down_{i}|)\cdot ((|Q_{i}|+|\down_{i}|)/(|S_i|+|Q_{i}|+|\down_{i}|)+3^k \cdot (|Q_{i}|+|\down_{i}|)/|S_i|))=O(m)$.
As we terminate the insertion algorithm when a rebuild occurs, the rebuild can occur at most once for each batch, which contributes $O(m)$ in expectation to the work and $O(\log (n+m) \log^* (n+m))$ \whp{} to the depth by Theorem~\ref{theorem:build}.

For the rest of the algorithm,
in terms of work, Line~\ref{line:grid-insert-short:insert}--\ref{line:grid-insert-short:check2} does work proportional to $O(\sum_i (|Q_i|+|down_i|))=O(m)$ across all the levels due to Lemmas~\ref{lemma:packing} and~\ref{lemma:sum-q}.
On Line~\ref{line:grid-insert-short:check3}, the number of points in the neighborhood $N_i(x,S_i')$ of each $x$ is upper bounded by $3^k$ since the points in $S_i'$ are sparse, therefore it takes $O(3^k\cdot m)=O(m)$ expected work.
Note that the work is amortized due to resizing the parallel dictionary when necessary.
In terms of depth, looking up and inserting points takes $O(\log^* (n+m))$ depth using the parallel dictionary. Therefore, all operations in Lines~\ref{line:grid-insert-short:insert}--\ref{line:grid-insert-short:check3} takes $O(\log^* (n+m))$ depth, and across all $O(\log (n+m))$ \whp{} rounds, the total depth is $O(\log (n+m)\log^* (n+m))$ \whp{}.
\end{proof}

\begin{lemma}\label{lemma:grid-update-delete}
We can maintain a sparse partition for a batch of $m$ deletions in $O(m)$ amortized work in expectation and $O(\log (n+m) \log^*(n+m))$ depth \whp.
\end{lemma} 

Now we describe the parallel updates of min-heaps $H_i$ associated with each level $i$ of the sparse partition.
Recall that $H_i$ contains the restricted distances $d_i^*(q)$ for $q\in S_i'$. By definition, $d_i^*(q)$ is the closest distance of $q$ to another point in $S_{i-l}'$ where $0\leq l\leq k$ ($k$ is the dimensionality). Therefore, following an update on $S_i'$, we need to update the $d_i^*(q)$'s in $H_{i+l}$ for $0\leq l\leq k$, and $q\in S_{i+l}'$. We use same example in Figure~\ref{fig:overview}~(left), where we denote the restricted distance of point $x$ as $(\textit{\textbf{x}},y)_i = d^*_i(x)= d(x,y)$, where $y\in \bigcup_{0\leq j\leq k} S_{i-j}'$ is another point that defines $x$'s closest distance.
As shown in Figure~\ref{fig:overview} (left), due to the insertion of the sparse point $g$ to $S_1$, entry $(\textit{\textbf{g}},b)_1$ is added to $H_1$. Some entries in $H_1$ are moved due to the point movements, e.g., $(\textit{\textbf{a}},e)_1$ from $H_1$ is moved and updated to $(\textit{\textbf{a}},f)_2$ in $H_2$ because $a$ has moved from $S_1$ to $S_2$, and $f$ is now closer.
Some entries are updated, e.g., $(\textit{\textbf{c}},d)_2$ is updated to $(\textit{\textbf{c}},f)_2$ in $H_2$ since the new point $f$ is closer to $c$ than $d$.

Our algorithm uses a new parallel batch-dynamic binary heap that we introduce in \cref{sec:heap}. For a batch-parallel binary heap of size $n$ and a batch update (a mix of inserts, deletes, and increase/decrease-keys) of size $m$, updating the heap takes $O(m\log ((n+m)/m))$ work and $O(\log (n+m))$ depth, and find-min takes $O(1)$ work. 
A simple implementation of our algorithm 
would execute all min-heap updates for an updated $S_i'$ before processing the updates for $S_j'$ where $j>i$ for insertions and $j<i$ for deletions. This level-by-level dependence leads to $O(\log^2(n+m))$ depth overall for the heap updates. In \full{\cref{appendix:heap-update},}{the full version of the paper,} we present an improved algorithm that breaks the level dependencies and enables updates to be performed to a min-heap as early as possible, leading to $O(\log(n+m))$ depth overall for the heap updates.  Together with Lemma~\ref{lemma:grid-update} and Lemma~\ref{lemma:grid-update-delete}, we obtain the following theorem, whose analysis we present in \full{\cref{appendix:overall-analysis}.}{the full version of the paper.}


\begin{theorem}\label{thm:overall}
Updating our data structure for a batch of $m$ insertions/deletions takes amortized $O(m(1+\log ((n+m)/m)))$ expected work and $O(\log (n+m) \log^*(n+m))$ depth \whp. 

\end{theorem}

Obtaining the closest pair from our data structure takes $O(1)$ work and depth. We simply call find-min on $H_i$ for $L-k\leq i \leq L$, and then take the overall minimum.

\section{Parallel Batch-Dynamic Binary Heap}\label{sec:heap}

One of the key components in parallelizing our closest pair algorithm is a parallel binary heap that supports batch updates (inserts and deletes) and find-min efficiently.
This heap allows us to perform the parallel construction in linear work and perform updates with low depth.
Our data structure may be of independent interest, since to the best of our knowledge, the only existing work on parallelizing a binary heap is on individual inserts or deletes~\cite{pinotti1995parallel}.

A binary heap is a complete binary tree, where each node contains a key that is smaller than or equal to the keys of its children.
Sequentially, the construction of a binary heap takes linear work, and each insert and delete takes $O(\log n)$ work~\cite{CLRS}.
The heap is represented as an  array, and uses relative positions within the array to represent child-parent relationships.
Sequentially, each insertion adds a new node at the end of the heap and runs \textsc{Up-Heap} to propagate the node up to the correct position in the heap. A deletion first swaps the node to delete with the node to the end of the heap, reduces the heap size by one, and then runs \textsc{Up-Heap} followed by \textsc{Down-Heap} (to propagate a node down to its correct position) for the node swapped to the middle of the heap. 

Central to our parallel batch-dynamic binary heap is a new parallel \heapify algorithm, that takes $m$ updates from a valid heap of $n$ elements, and returns another valid \full{heap (the pseudocode can be found in \cref{alg:heap} in \cref{appendix:heap}).}{heap.} It runs in two phases: the first phase works on increase-key updates, and the second phase on decrease-key updates. In both phases, we first use parallel integer sorting~\cite{RR89,vishkin10} to categorize all updates based on the level where the update belongs.
Simply running the \textsc{Up-Heap} and \textsc{Down-Heap} calls for the different updates in parallel does not achieve work-efficiency and low depth, and also leads to potential data races.
Therefore, we \emph{pipeline} each level of the \textsc{Up-Heap} and \textsc{Down-Heap} procedure. Specifically, in the first phase, once the first swap for the \textsc{Down-Heap} in level $i$ is finished, we can immediately start the \textsc{Down-Heap} on level $i-1$, instead of waiting for the \textsc{Down-Heap} in level $i$ to completely finish (the root is at level $0$, and level numbers increase going down). 
The swaps in the \textsc{Down-Heap} calls from level $i-1$ will never catch up with the swaps from level $i$.
Pipelining the second phase with \textsc{Up-Heap} is more complicated. Our parallel \textsc{Up-Heap} is run in a level-synchronous manner from the top level down to the bottom level.
For each node on each level, both of its children may want to swap with the parent for having a larger value. In the parallel algorithm, we only make the child with the smaller value swap with the parent and continue its update to the upper levels, while the update for the other child terminates. We prove in \full{\cref{appendix:heap}}{the full paper} that our parallel \textsc{Heapify} algorithm takes $O(m(1+\log (n/m)))$ work and $O(\log n)$ depth.

We now explain how to perform batch insertions and deletions.
A batch of $m$ insertions to a binary heap of size $n$ can be implemented using decrease-keys.
We first add the $m$ elements to end of the heap with keys of $\infty$. Then, we decrease the keys of these $m$ elements to their true values and run the parallel \heapify algorithm.
A batch of $m$ deletions can be processed similarly, but the deletions will generate ``holes'' in the tree structure, and so we need an additional step to fill these holes first.
We pack the last $m$ elements in the heap based on whether they are deleted.
Then, we use them to fill the rest of the empty slots by deletions, and run the parallel \heapify algorithm.
Hence, batch insertions and deletions take $O(m(1+\log ((n+m)/m)))$ work and $O(\log(n+m))$ depth.

We provide more details about our data structure and prove the following theorem in \full{\cref{appendix:heap}.}{the full version of the paper.}

\begin{theorem}\label{thm:heap}
For a batch-parallel binary heap of size $n$ and a batch update (a mix of inserts, deletes, and increase/decrease-keys) of size $m$, updating the heap takes $O(m(1+\log ((n+m)/m)))$ work and $O(\log (n+m))$ depth, and find-min takes $O(1)$ work.
\end{theorem}
\section{Implementations}\label{sec:implementation}

\myparagraph{Simplified Data Structure}
While the sparse partition maintains $(S_i,S_i',p_i,q_i,d_i)$ and $H_i$ for each level $1\leq i\leq L$, we found that implementing $S_i'$ and its associated heap $H_i$ on every level was inefficient in practice. We found it more efficient to only maintain $(S_i,p_i,q_i,d_i)$ for $1\leq i\leq L$, and one heap $H^*$ that stores the closest neighbor distances for all $q$ in $S_j$, where $j= L-\lceil \log_3 {2\sqrt{k}} \rceil$.
When $L$ changes due to insertion or deletion, we recompute $j$ and rebuild $H^*$ if necessary. We prove in \full{\cref{sec:appendix-simplified}}{the full version of the paper} that $H^*$ contains the closest pair. Our implementation uses the parallel heap from~\cite{sun2019parallel}.
Additionally, we compute $S_i'$  from $S_i$ on the fly when needed.

\myparagraph{Neighborhood Search}
Some of the work bounds are exponential in the dimensionality $k$, e.g., a grid's box neighborhood is of size $3^k$. For $k\geq 5$, the straightforward implementation is inefficient due to a large constant overhead in the work.
Hence, we implement a parallel batch-dynamic \kdt for $k\geq 5$.
This is because performing a range query on the tree works better in practice, as it only needs to traverse the non-empty boxes in the neighborhood instead of all boxes.
Our dynamic \kdt is a standard spatial median \kdt \cite{bentley1975multidimensional}, augmented with the capability for parallel batch updates.
Each internal node maintains metadata on the points in its subtree, which are partitioned by a spatial median along the widest dimension.
The points are only stored at leaf nodes. We flatten a subtree to a single leaf node when it contains at most 16 points.

The tree supports batch insertion by first adding the batch to the root, and then traversing down multiple branches of the tree in parallel.
At each internal node, we partition the inserted batch by the spatial median stored at the node, and modify its metadata, such as the point count and the coordinates of its bounding box. At each leaf node, we directly modify the metadata and store the points.
The tree supports batch deletions by modifying the metadata, and marking the deleted points at the leaves as invalid. We manage the memory periodically to free up the invalid entries. 

\myparagraph{Static Algorithms}
In addition to our batch-dynamic closest pair algorithm, 
we implement several sequential and parallel algorithms for the static closest pair problem. 
As far as we know, this paper presents the first experimental study of parallel algorithms for static closest pair.
We implement a parallel divide-and-conquer algorithm by Blelloch and Maggs~\cite{blelloch2010parallel},
a simplified and parallel version of Rabin's algorithm~\cite{rabin1976probabilistic} that we designed,
a parallel version of Khuller and Matias's~\cite{khuller1995simple} sieve algorithm that we designed,
and a parallel randomized incremental algorithm by Blelloch et al.~\cite{BGSS16}. 
We explain more details about these static algorithms and their implementations in \full{\cref{appendix:static}.}{the full paper.}

\section{Experiments}\label{sec:experiment}
\myparagraph{Algorithms Evaluated}
We evaluate our parallel batch-dynamic algorithm by benchmarking its performance on batch insertions (\defn{dynamic-insert}) and batch deletions (\defn{dynamic-delete}). We also evaluate the four static implementations described in Section~\ref{sec:implementation}, which we refer to as \defn{divide-conquer}, \defn{rabin}, \defn{sieve}, and \defn{incremental}. 
 In addition, we implement and evaluate sequential versions of all of our algorithms that do not have the overheads of parallelism.
Our implementations use the Euclidean metric ($L_2$-metric).

\myparagraph{Data Sets}
We use the synthetic seed spreader (SS) data sets produced by the
generator in~\cite{GanT17}. It produces points generated
by a random walk in a local neighborhood, but jumping to a random
location with some probability. \defn{SS-varden}
refers to the data sets with variable-density
clusters.
We also use a synthetic data set called
\defn{Uniform}, in which points are distributed uniformly at
random inside a bounding hyper-cube with side length $\sqrt{n}$, where
$n$ is the total number of points. The points have double-precision floating-point
values. We generated the
synthetic data sets with 10 million points for
dimensions $k=2,3,5,7$.
We name the data sets in the format of \defn{Dimension-Name-Size}.
We also use the following real-world data sets:
\defn{7D-Household-2M}~\cite{UCI} is a 7-dimensional data set containing household sensor data with $2,049,280$ points excluding the date-time information;
\defn{16D-Chem-4M}~\cite{fonollosa2015reservoir, CHEMURL} is a 16-dimensional data set with $4,208,261$ points containing chemical sensor data; and
\defn{3D-Cosmo-298M}~\cite{Kwon2010} is a 3-dimensional astronomy data set with $298,246,465$
points.

\myparagraph{Testing Environment} Our experiments are run on
an \texttt{r5.24xlarge} instance on Amazon EC2.  The machine has 2
$\times$ Intel Xeon Platinum 8259CL CPU (2.50 GHz) CPUs for a total of 48
 cores with two-way hyper-threading, and 768 GB of RAM. 
By default, we use all cores
with hyper-threading.  
We use the \texttt{g++} compiler (version 7.5)
with the \texttt{-O3} flag, and use Cilk Plus for parallelism~\cite{cilkplus}.
We use the \defn{-48h} and \defn{-1t} suffixes in our algorithm names to denote the 48-core with hyper-threading and single-threaded times, respectively.
We allocate a maximum of 2 hours for each test, and do not report times for tests that exceed this limit.

\begin{figure}[]
\centering
  \begin{minipage}[t]{0.49\textwidth}
    \includegraphics[width=\textwidth]{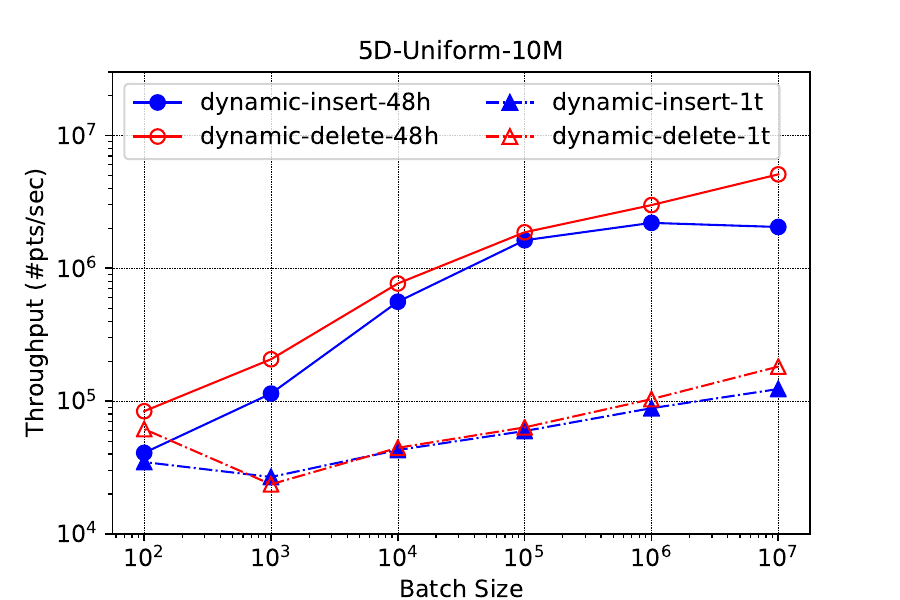}
  \end{minipage} %
  \begin{minipage}[t]{0.49\textwidth}
    \includegraphics[width=\textwidth]{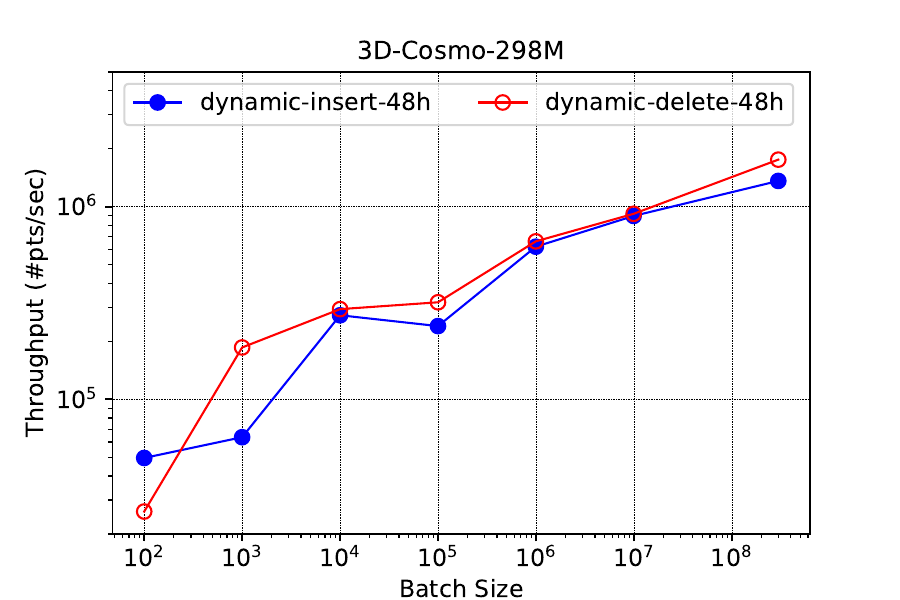}
  \end{minipage} %
  \caption{Plots of throughput vs.\ batch size in log-log scale for our parallel batch-dynamic algorithm on 5D-Uniform-10M and 3D-Cosmo-298M. 
  The algorithm on 48-cores with hyper-threading and 1 thread has a suffix of "48h" and "1t", respectively.
For 3D-Cosmo-298M, we omit the 1-thread times as the experiments exceeded our time limit.}\label{fig:throughput}
\end{figure} %

\myparagraph{Influence of Batch Size on Throughput}
In this experiment, we evaluate our batch-dynamic algorithm by measuring their throughput as a function of the batch size. For insertions, we insert batches of the same size until the entire data set is inserted. For deletions, we start with the entire data set and delete batches of the same size until the entire data set is deleted.
We compute throughput by the number of points processed per second.
We vary the batch size from $100$ points to the size of the entire data set.
Our parallel batch-dynamic algorithm achieves a throughput of up to $1.35\times 10^7$ points per second for insertion, and up to $1.06\times 10^7$ for deletion, under the largest batch size. On average, it achieves $1.75\times 10^6$ for insertion and $1.94\times 10^6$ for deletion across all batch sizes.
We show plots of throughput vs.\ batch size for 5D-Uniform-10M and 3D-Cosmo-298M in Figure~\ref{fig:throughput}.
We see that the throughput increases with larger batch sizes because of a lower relative overhead of traversing the sparse partition data structure, and the availability of more parallelism. 
\full{For completeness, the average update throughput for both insertions and deletions under varying batch sizes are shown in Table~\ref{table:throughput} in \cref{appendix:experiment}.}{}

\begin{figure}[]
\centering
  \begin{minipage}[t]{0.49\textwidth}
    \includegraphics[width=\columnwidth]{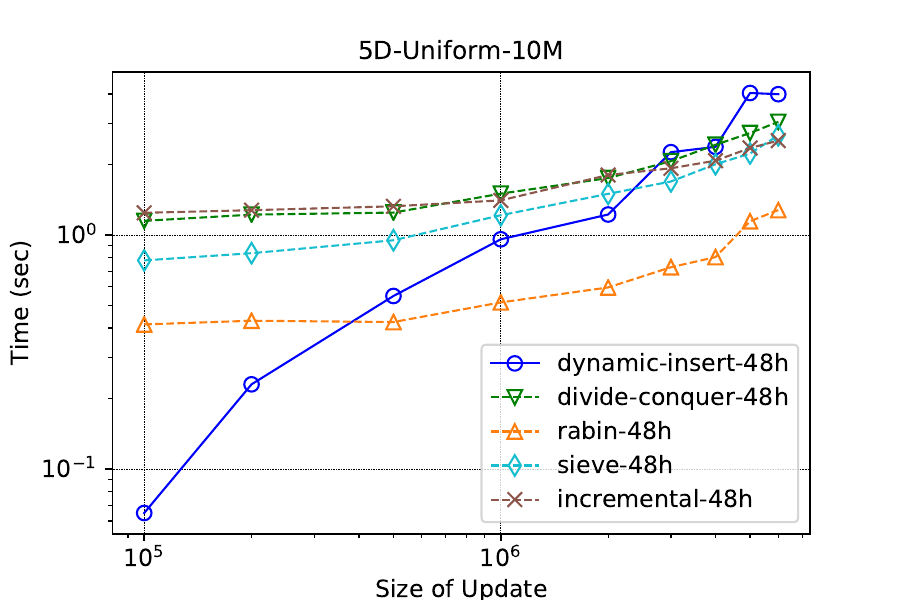}
  \end{minipage} %
  \begin{minipage}[t]{0.49\textwidth}
    \includegraphics[width=\columnwidth]{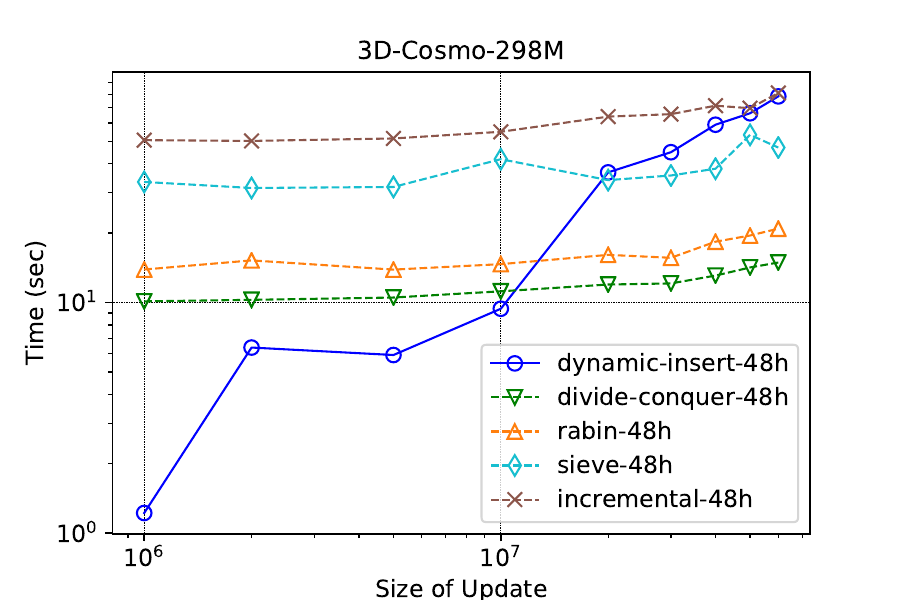}
  \end{minipage} %
  \caption{Plots of running time (in seconds) vs.\ insertion batch size for the dynamic and static methods using 48 cores with hyper-threading on 5D-Uniform-10M and 3D-Cosmo-298M. The plots are in log-log scale.}\label{fig:advantage}
\end{figure} %

\myparagraph{Efficiency of Batch Insertions}
We evaluate the performance of dynamic batch insertion vs.\ using a static algorithm to recompute the closest pair.
Specifically, we simulate a scenario where given the data structure storing the closest pair among $c$ data points, we perform an insertion of $b$ additional points. We compare the time taken by the dynamic algorithm to process one batch insertion of size $b$, vs.\ that of a static algorithm for recomputing the closest pair for all $c+b$ points.
We set $c$ to contain 40\% of the data set and vary $b$.
Figure~\ref{fig:advantage} shows the running time as a function of $b$ for 5D-Uniform-10M and 3D-Cosmo-298M.
For 5D-Uniform-10M, we see that our batch-dynamic algorithm outperforms the fastest among the static algorithms when the insertion batch size is smaller than 500,000.
For 3D-Cosmo-298M, we see that the dynamic method outperforms the fastest static algorithm when the insertion batch is smaller than 10 million.
In general, both the static and dynamic algorithms require more time to process the updates when the batch size is larger. The dynamic algorithm is much more advantageous for small to moderate batch sizes.

\begin{figure}[]
\centering
  \begin{minipage}[t]{0.49\textwidth}
    \includegraphics[width=\columnwidth]{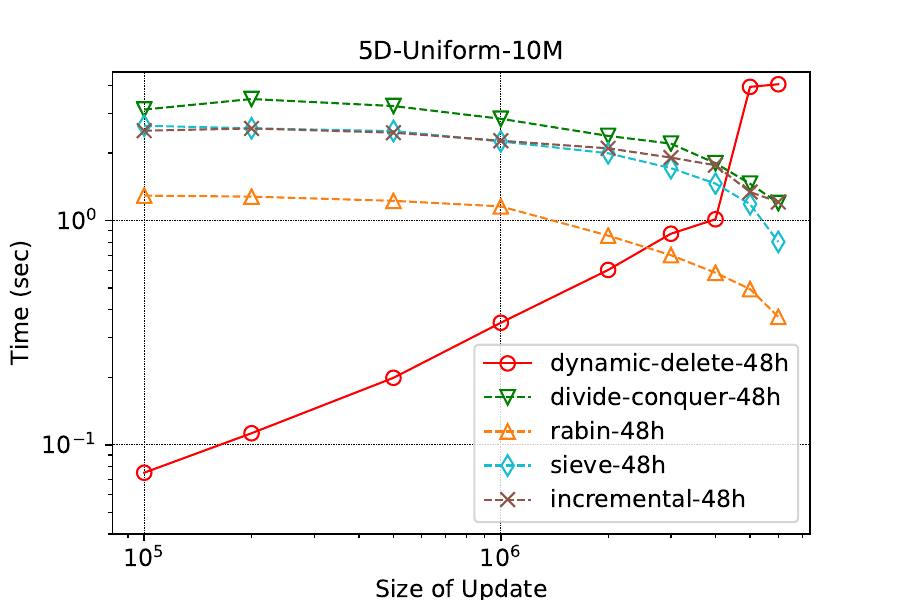}
  \end{minipage} %
  \begin{minipage}[t]{0.49\textwidth}
    \includegraphics[width=\columnwidth]{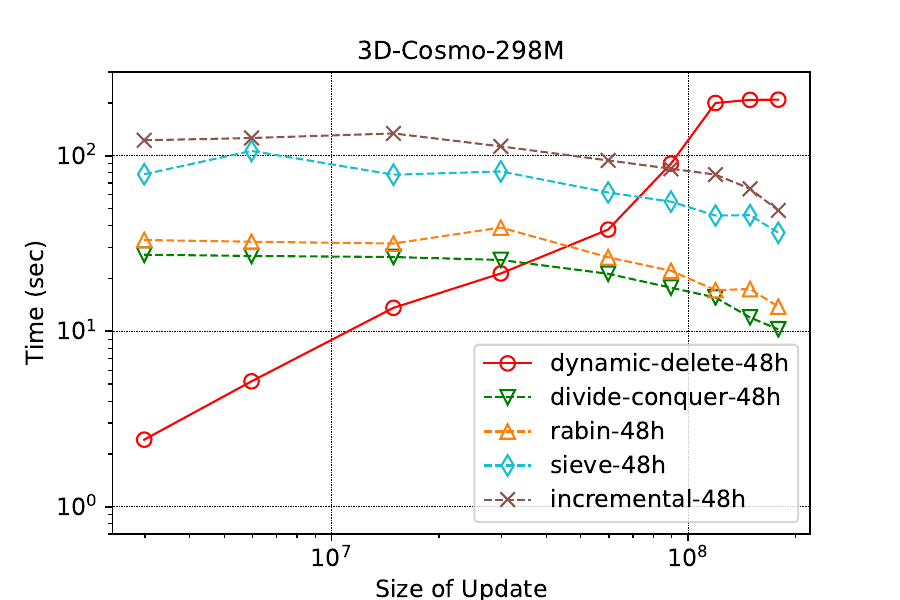}
  \end{minipage} %
  \caption{Plots of running time (in seconds) vs.\ deletion batch size for the dynamic and static methods using 48 cores with hyper-threading on 5D-Uniform-10M and 3D-Cosmo-298M. The plots are in log-log scale.}\label{fig:advantage-deletion}
\end{figure}

\myparagraph{Efficiency of Batch Deletions}
We evaluate the performance of dynamic batch deletion vs.\ using a static algorithm to recompute the closest pair. In this experiment, we are given the closest pair of all $n$ points in the data set, and perform a deletion of $b$ points. We compare the time taken for the dynamic algorithm to process one batch deletion of size $b$, vs.\ that of a static algorithm for recomputing the closest pair for the $n-b$ remaining points.
Figure~\ref{fig:advantage-deletion} shows the running time vs.\ deletion batch size for 5D-Uniform-10M and 3D-Cosmo-298M.
For 5D-Uniform-10M, the dynamic algorithm outperforms the fastest static algorithm when the batch size is less than 3 million.
For 3D-Cosmo-298M, the dynamic algorithm outperforms the static algorithm when the batch size is less than 60 million.
In general, our dynamic algorithm requires more time to process the update when the batch size is larger, while the converse is true for the static algorithms. 

Compared to the fastest static algorithms on our data sets, we find that it is faster to use our dynamic algorithm for batch sizes of up to 20\% of the data set.

\begin{table}[t]
\centering
\setlength{\tabcolsep}{4pt}
\begin{tabular}{c|ccc|ccc|ccc|ccc}
                & \multicolumn{3}{c|}{Divide-Conquer} & \multicolumn{3}{c|}{Rabin} & \multicolumn{3}{c|}{Sieve} & \multicolumn{3}{c}{Incremental} \\
                & Seq        & 1t        & 48h     & Seq      & 1t     & 48h & Seq      & 1t     & 48h & Seq       & 1t       & 48h    \\ \hline
2D-Uniform-10M  & 9.54       & 9.62       & 0.24      & 11.2    & 11.6   & 0.28  & 23.3    & 24.5   & 0.81  & 22.1     & 17.7     & 1.02     \\
3D-Uniform-10M  & 24.9      & 25.2      & 0.66      & 28.4    & 30.5   & 0.78  & 60.3    & 60.6   & 1.82  & 50.5     & 46.2     & 2.50     \\
5D-Uniform-10M  & 101     & 136     & 3.04      & 25.3    & 28.4   & 1.28  & 56.7    & 60.6   & 2.63  & 49.2     & 50.3     & 2.40     \\
7D-Uniform-10M  & 561     & 618     & 14.7     & 81.7    & 82.8   & 1.70  & 124   & 135  & 4.24  & 93.7     & 106    & 4.58     \\
2D-SS-varden-10M   & 7.58       & 8.95       & 0.23      & 10.5    & 11.2   & 0.26  & 22.2    & 22.8   & 0.94  & 23.4     & 17.5     & 1.11     \\
3D-SS-varden-10M   & 17.3      & 19.1      & 0.51      & 28.4    & 29.1   & 0.77  & 58.4    & 58.3   & 1.68  & 48.7     & 43.1     & 1.97     \\
5D-SS-varden-10M   & 24.9      & 33.4      & 0.82      & 22.6    & 26.1   & 1.43  & 47.2    & 49.3   & 2.58  & 40.4     & 41.7     & 2.44     \\
7D-SS-varden-10M   & 43.1      & 50.3      & 1.33      & 33.1    & 34.0   & 1.61  & 64.4    & 70.9   & 3.00  & 43.4     & 48.0     & 2.53     \\
7D-Household-2M & 342     & 392     & 13.4     & 7.23     & 7.70    & 0.40  & 15.9    & 18.1   & 0.73  & 13.8     & 15.7     & 0.94     \\
16D-Chem-4M     & 315     & 499     & 202    & 38.3    & 39.8   & 1.38  & 88.2    & 96.7   & 2.68  & 59.1     & 70.8     & 3.91     \\
3D-Cosmo-298M   & 750     & 747     & 20.7     & 1243  & 1625 & 31.6 & 3383  & 2819 & 70.6 & 3456   & 2629   & 104  
\end{tabular}

\caption{Running times (in seconds) of static algorithms. "Seq" denotes the sequential implementation. "1t" and "48h" denote the parallel implementation run on 1 thread and 48 cores with hyper-threading, respectively.
}\label{table:static}
\end{table}

\myparagraph{Static Methods}
We evaluate and compare the static algorithms and present all detailed running times in Table~\ref{table:static}.
Among the four parallel static algorithms, Rabin's algorithm is on average 7.63x faster than the rest of the algorithms across all data sets.
The divide-and-conquer, sieve, and the incremental algorithms are on average 17.86x, 2.29x, and 2.73x slower than Rabin's algorithm, respectively.
The divide-and-conquer algorithm actually achieves the fastest parallel running time on 7 out of the 11 data sets. However, it is significantly slower for most of the higher dimensional data sets, due to its higher complexity with increased dimensionality.
The sieve algorithm and the incremental algorithm, though doing the same amount of work in theory as Rabin's algorithm, have higher constant factor overheads.

\myparagraph{Parallel Speedup and Work-Efficiency}
We measure the parallel speedups of our implementations by dividing the 1-thread time by the 48-core with hyper-threading time.
Our parallel batch-dynamic algorithm achieves up to 38.57x self-relative speedup (15.10x on average across all batch sizes), averaging over both insertions and deletions. 
Our static implementations achieve up to 51.45x speedup (29.42x on average). Specifically, the divide-and-conquer algorithm, Rabin's algorithm, the sieve algorithm, and the incremental algorithm achieve average self-relative speedups of 35.17x, 33.84x, 29.22x, and 19.45x, respectively; in addition, they achieve an average speedup of 19.10x, 23.56x, 10.56x, and 9.23x, respectively, over the fastest serial algorithm for each data set.

Our parallel implementations when run on one thread demonstrate modest overheads over their sequential counterparts.
Our parallel batch-dynamic algorithm running on 1 thread has only 1.13x lower throughput on average over our sequential implementation of the algorithm. 
For the static algorithms, the parallel divide-and-conquer, Rabin's, sieve, and incremental algorithms running on 1 thread are only 1.18x, 1.08x, 1.04x, and 1.00x slower on average, respectively, than their corresponding sequential algorithms.


\bibliography{ref}

\full{
\appendix

\section{Extended Preliminaries}\label{appendix:prelims}

In this section, we summarize the parallel primitives used in this paper. We also summarize all notations used in Table~\ref{table:notation}.

\begin{table}[t]
\centering
  \begin{tabular}{|c | p{9cm} |}
\hline
\textbf{Notation} & \multicolumn{1}{c|}{\textbf{Definition}} \\ \hline
$k$ & Dimensionality of the data set.\\ \hline
$S$ & Point data set $\{p_1, p_2, \ldots, p_{n}\}$ in $\mathbb{R}^k$. \\ \hline
$n$ & Size of $S$ ($|S|$).\\ \hline
$m$ & Size of a batch update.\\ \hline
$d(p,q)$ & Distance between points $p,q\in S$. \\ \hline
$\delta(S)$ & $\min\{d(p,q):p,q\in S, p\neq q\}$, i.e., the distance of the closest pair in set $S$. \\ \hline
$d(p,S)$ & $\min\{d(p,q): q\in S \setminus p\}$, i.e., the distance of $p$ to its nearest neighbor in set $S$. \\ \hline
$d_i^*(p)$ & The restricted distance of point $p$, $d_i^*(p) := d(p,S_{i-k}'\cup S_{i-k+1}'\cup \ldots \cup S_i')$. \\ \hline
$(S_i,S_i',p_i,q_i,d_i)$ & The 5-tuple representing each level of the sparse partition data structure, where $S_i$ and $S_i'$ are point sets, $p_i$ is the pivot point, $q_i$ is the closest point to $p_i$ in $S_i$, and $d_i:=d(p_i,q_i)$. \\ \hline
$G_i$ & The grid structure with box side length $d_i/6k$ at level $i$ of the sparse partition. \\ \hline
$H_i$ & The parallel heap associated with level $i$ of the sparse partition. \\ \hline
$L$ & The number of levels in the sparse partition. \\ \hline
$b_i(p)$ & The box containing $p$ on level $i$. \\ \hline
$b^{\sigma}_i(p)$ & The box with a offset of $\sigma$ relative to $b_i(p)$. $\sigma$ is a $k$-tuple over $\{-1,0,1\}$, where the $j$'th component indicates the relative offset in the $j$'th dimension. \\ \hline
$BN_i(p)$ & The box neighborhood of $p$, i.e., the collection of the $3^k$ boxes bordering and including the box containing $p$ on level $i$. \\ \hline
$BN^{\sigma}_i(p)$ & The partial box neighborhood of $p$, i.e., the intersection of $N_i(p)$ with the boxes bordering and including $b^{\sigma}_i(p)$. \\ \hline
$N_i(p,S)$ & The neighborhood of $p$ in set $S$, i.e., the set of points in $S\setminus p$ contained in $N_i(p)$. \\ \hline

\end{tabular}
  \caption{Summary of Notation.}\label{table:notation}
\end{table}

\defn{Prefix sum} takes as input a sequence $[a_1, a_2, \ldots , a_{n}]$, an
associative binary operator $\oplus$, and an identity $i$, and returns
the sequence $[i, a_1, (a_1 \oplus a_2), \ldots , (a_1 \oplus a_2
  \oplus \ldots \oplus a_{n-1})]$ as well as the overall sum of the
elements.  \defn{Filter} takes an array $A$ and a
predicate function $f$, and returns a new array containing $a \in A$
for which $f(a)$ is true, in the same order that they appear in $A$.
Both prefix sum and filter can be implemented in $O(n)$ work and $O(\log n)$
depth~\cite{JaJa92}. We use a parallel \defn{minimum} algorithm, which computes the minimum of $n$ points in $O(n)$ expected work and $O(1)$ depth \whp{}~\cite{vishkin10}. 
We use parallel dictionaries, which support $n$ insertions, deletions, or lookups in $O(n)$ work and $O(\log^*n)$ depth \whp{}~\cite{Gil91a}. Finally, we use \defn{integer sorting}  on $n$ keys in the range $[0,\ldots,O(\log n)]$, which takes $O(n)$ work and $O(\log n)$ depth~\cite{RR89,vishkin10}.

\section{Restricted Distance}\label{appendix:golin}

The \defn{restricted distance} $d^*_i(p)$~\cite{golin1998randomized} is the closest pair distance from point $p$ to any point in $\bigcup_{0\leq j \leq k} S_{i-j}'$, and defined as $d^*_i(p):= \min\{d_i, d(p, S_{i-k}'\cup S_{i-k+1}' \cup \ldots \cup S_i')\}$, where $p\in S_i'$.
Golin et al.\ show that $\delta(S) = \min_{L-k\leq i\leq L} \min_{p\in S_i'} d_i^*(p)$, meaning that the closest pair can be found by taking the minimum among the restricted distance pairs for all points in last $k+1$ levels of $S_i'$.
For completeness, we present the following lemmas and theorem from Golin et al.~\cite{golin1998randomized} to show that $\delta(S) = \min_{L-k\leq i\leq L} \min_{p\in S_i'} d_i^*(p)$.

\begin{lemma}\label{lemma:rd}
$d_i^*(p)>d_i/6k$ for $p\in S_i$. 
\end{lemma}

\begin{proof}
Let $1\leq j\leq i$ and let $q\in S_j'$, where $q$ is an arbitrary point. Since $p\in S_j$, it follows from properties of the sparse partition (\cref{sec:golin:sp}) that $d(p,q)\geq d(q,S_j) > d_j/6k\geq d_i/6k$. Therefore the distance from $p$ to its closest pair $d_i^*(p)>d_i/6k$.
\end{proof}

\begin{lemma}\label{lemma:rd2}
$d_L/6k < \delta(S)\leq d_L$. 
\end{lemma}

\begin{proof}
Let $\delta(S) = d(p,q)$ for some $p\in S_i'$ and $q\in S_j'$, and without loss of generality $i\leq j$. It follows from the definition of the sparse partition that $p,q\in S_i$, and we have $d(p,q)=d(p,S_i)>d_i/6k$, hence $d(p,q)>d_L/6k$.

$\delta(S)\leq d_L$ obviously holds since $d_L$ is the distance between two points. 
\end{proof}

\begin{theorem}\label{thm:rd}
$\delta(S) = \min_{L-k\leq i\leq L} \min_{p\in S_i'} d_i^*(p)$.
\end{theorem}

\begin{proof}

Since the restricted distance is the distance between two points, we have that $\delta(S)\leq \min_{1\leq i\leq L} \min_{p\in S_i'} d^*_i(p)$.
Let $\delta(S) = d(p,q)$ for some $p\in S_i'$ and $q\in S_j'$. Assume without loss of generality that $j\leq i$, and it is obvious that $d(p,q)=d(p,\bigcup_{h\leq i} S_h')\geq d_i^*(p)$. Therefore $\delta(S)\geq \min_{1\leq i\leq L} \min_{p\in S_i'} d_i^*(p)$, and hence $\delta(S) = \min_{1\leq i\leq L} \min_{p\in S_i'} d_i^*(p)$.

We now prove that $i$ cannot be less than $L-k$.
By Lemma~\ref{lemma:rd}, we have $\min_{p\in S_i'} d^*_i(p)> d_i/6k$. We also know from Lemma~\ref{lemma:rd2}, and from the properties of the sparse partition ($d_{i+1}\leq d_i/3$), that for $i<L-k$, $d_i/6k\geq d_{L-k-1}/6k \geq (3^{k+1}/6k)\cdot d_L > d_L \geq \delta(S)$.
Therefore, $\delta(S)\geq \min_{L-k\leq i\leq L} \min_{p\in S_i'} d_i^*(p)$.
\end{proof}

The sequential algorithm~\cite{golin1998randomized} computes the restricted distance for each point in $S_i'$, and stores them in min-heaps $H_i$, for $1\leq i\leq L$.
To obtain the closest pair, we simply read the minimum in $H_i$ for $L-k\leq i \leq L$ to obtain $k+1$ values, and then take the overall minimum. This takes $O(1)$ work and depth.
\section{Parallel Construction Analysis}\label{appendix:construction}

In this section, we present analysis for the parallel construction algorithm in Section~\ref{sec:par-alg} and Algorithm~\ref{alg:build-short}.
For each call to \textsc{Build}, given $O(|\s{i}|)$ points, insertions to the parallel dictionary take $O(|\s{i}|)$ work and $O(\log^* |\s{i}|)$ depth \whp.
Line~\ref{codeline:build-short:pivot} computes the distance of $p_i$ to each $q\in S_i$, taking $O(|S_i|)$ work and $O(1)$ depth. Then we obtain $q_i$ via a parallel minimum computation in the same work and depth.
Checking the sparsity of points takes $O(|\s{i}|)$ work and $O(1)$ depth, since each point checks $3^k$ boxes to see whether any are non-empty 
(Line~\ref{codeline:build-short:si}). Therefore,
except for the cost of Line~\ref{codeline:build-short:spawn} and the recursive call on Line~\ref{codeline:build-short:recur}, each call to \textsc{Build} takes $O(|\s{i}|)$ expected work and $O(\log^* |\s{i}|)$ depth \whp{}.
Line~\ref{codeline:build-short:spawn} creates a parallel heap with $O(|\sprime{i}|)$ entries, which takes $O(|\sprime{i}|)$ work and $O(\log |\sprime{i}|)$ depth, which we prove in \cref{appendix:heap}. 
Since $\sum{|\s{i}|}=O(n)$~\cite{golin1998randomized}, the total work across all calls to \textsc{Build} is hence $O(n)$ in expectation.
Since our heap is of linear size, the
total space usage of our data structure is also 
$O(n)$ in expectation. 

We now prove that the algorithm has polylogarithmic depth by first showing Lemma~\ref{lemma:rounds}.

\begin{lemma}\label{lemma:rounds}
Algorithm~\ref{alg:build-short} makes $O(\log n)$ calls to \textsc{Build} \whp, and the sparse partition has $O(\log n)$ levels \whp.
\end{lemma}

\begin{proof} 
We show that with at least $1/2$ probability, $|S_{i+1}|\leq |S_{i}|/2$.
Consider relabeling the points in $S_i=\{r_1, r_2, \ldots, r_{|S_i|}\}$ such that $d(r_1, S_i) \leq d(r_2, S_i) \leq \ldots \leq d(r_{|S_i|}, S_i)$. 
If we pick the pivot $p_i=r_j$, then for every $r_k$ with $k>j$, we have $d(r_j,S_i)\leq d(r_k,S_i)$, and so $r_k$ is not in $S_{i+1}$. Since the pivot is chosen randomly, it can be any $r_k$ with equal probability, and with $1/2$ probability $k \leq |S_i|/2$, implying that $|S_{i+1}|\leq |S_i|/2$.

By definition, we have at most $\log_2 n$ levels where the set size decreases by at least a half. We define such levels as \emph{good levels}.
We want to analyze the number of pivots needed to be chosen until we have $\log_2n$ good levels, at which point there are no additional levels in the data structure. Let $X$ be a random variable denoting the number of levels needed until we see  $\log_2n$ good levels. $X$ follows a negative binomial distribution with parameters $k=\log_2n$ successes, and $p=1/2$ probability of success.  

Since pivots are chosen independently across the levels, we can use the following form of the Chernoff bound for variables following a negative binomial distribution~\cite{harvey15}:
$$\Pr[X> L] \leq e^{-\delta^2k/(2(1-\delta))} \text{, where\ } 0<\delta<1 \text{\ and\ } L=\frac{k}{(1-\delta)p}.$$

If we set $\delta=7/8$, we have $L=16\log_2 n$ and
$$\Pr[X> 16\log_2 n] \leq e^{-(7/8)^2 \log_2n/(2(1-7/8))} = e^{-(49/16)\log_2 n} < n^{-3}.$$




This means that the sparse partition has no more than $16\log_2 n=O(\log n)$ levels \whp, which also bounds the number of recursive calls to \textsc{Build} in Algorithm~\ref{alg:build-short}.
\end{proof}

As mentioned earlier, the depth of each call to \textsc{Build} is $O(\log^* n)$ \whp, excluding the cost for heap construction and the recursive call. Since the heap insertions are asynchronous, they take a total of $O(\log n)$ depth. Therefore, the total depth of Algorithm~\ref{alg:build-short} is $O(\log n\log^* n)$ \whp. This gives Theorem~\ref{theorem:build}.

\section{Parallel Insertion}\label{appendix:insertion}

\begin{figure}[t]
\center
  \begin{minipage}[t]{0.49\textwidth}
    \includegraphics[width=0.5\textwidth]{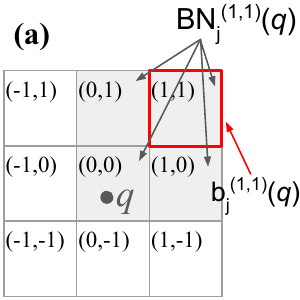}
  \end{minipage}
  \begin{minipage}[t]{0.49\textwidth}
    \includegraphics[width=0.5\textwidth]{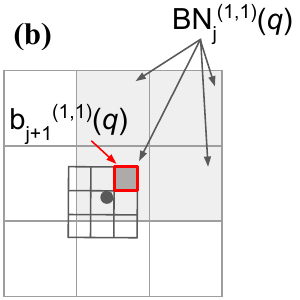}
  \end{minipage}
  \caption{(a) shows $BN_j(q)$, the box neighborhood of point $q$. Box $b_j^{(1,1)}(q)$ is marked in red. The partial box neighborhood $BN_j^{(1,1)}(q)$ is shaded. (b) shows that $b_{j+1}^{(1,1)} (q)$ is spatially contained in $BN_j^{(1,1)}(q)$.}\label{fig:packing}
\vspace{1em}
\end{figure}

This section presents the missing proofs of \cref{lemma:packing} and \cref{lemma:sum-q} from \cref{sec:par-alg}.

\packing*
\begin{proof}
We want to prove that the number of points moved across sparse partitions for the insertion of $Q$ with $m$ points is $O(m)$, i.e., $|\bigcup_{1\leq i\leq L} \down_i| \leq m\cdot 3^k = O(m)$.
Our proof shares some notation with Golin et al.~\cite{golin1998randomized}. For a level $i$ and a point $q\in Q$, we let  $\down_i(q)=\down_i\cap N_i(q,S_i')$.
Let $b_i(q)$ denote the box that contains point $q$. Let the \defn{box neighborhood} of $q$ in $G_i$, denoted by $BN_i(q)$, consist of $b_i(q)$ itself and the collection of $3^k-1$ boxes bordering $b_i(q)$.
We number the $3^k$ boxes in $BN_i(q)$ as a $k$-tuple over values $\{-1,0,1\}$, where the $j$'th component indicates the relative offset of the box with respect to $b_i(q)$ in the $j$'th dimension. We denote the box with a relative offset of $\sigma$ with respect to $b_i(q)$ as $b_i^\sigma(q)$, where $\sigma$ is the $k$-tuple (see \cref{fig:packing}a).
We further define the \defn{partial box neighborhood} of a point $q$, denoted by $BN_i^\sigma(q)$, as the set of boxes in $BN_i(q)$ 
that intersect with the boxes bordering on and including $b_i^\sigma(q)$ (see \cref{fig:packing}a).

We first show that $\sum_{l> j} |\down_{l}(q)| \leq 3^k$ for a single insertion $q$ with respect to level $j$.
Let $x\in \down_{j+1}(q)$ for some level $j$. By definition, $x$ is in a box of $BN_j(q)$ for some $q\in Q$ and $N_j(x,S)=\emptyset$. 
Therefore, the boxes in the partial box neighborhood $BN_j^\sigma (q)$ contain no points other than $x$ itself.
Suppose $x$ is in the box $b_j^{\sigma}(q)$. Now consider the box neighborhood of $x$ in level $l$ where $l>j$, and for the sake of contradiction, suppose there is some other point $y\in b_l^\sigma (q)$. Since $d_l\leq d_{j+1}\leq d_j/3$ by properties of the sparse partition, $b_l^\sigma (q)$ is spatially contained in $BN_j^\sigma(q)$, and hence we have $y\in BN_j^\sigma (q)$, which is a contradiction (see \cref{fig:packing}b).
Therefore, for any level $l>j$, there cannot be any point in $\down_{l}(q)$ with signature $\sigma$ except for $x$. 
For any $q\in Q$, since the number of partial box neighborhoods $BN_l^{\sigma}(q)$ across all value of $\sigma$ that do not share any points with each other 
is at most $3^k$, we have that $\sum_{l> j} |\down_{l}(q)| \leq 3^k$. 

We have shown that the number of points moved by a single insertion across levels is at most $3^k$, and now we extend the argument to batch insertion.
Consider another point $p\in Q$ inserted in the same batch as $q$. If $BN_j(p)$ and $BN_j(q)$ do not share any boxes, then the argument above holds for $p$ and $q$ independently.
We are concerned with the case where $BN_j(p)$ and $BN_j(q)$ overlap.
However, analyzing $p$ and $q$ separately can only overcount the point $x$ in the above argument if it appears in the partial box neighborhoods of both $BN_j^{\sigma'} (p)$ and  $BN_j^{\sigma} (q)$ for some $\sigma$ and $\sigma'$.
This argument can be extended to an arbitrary subset of $Q$ beyond $p$ and $q$ in a similar manner. Therefore, we can analyze each point separately to get an upper bound of $\sum_{q\in Q, 1\leq i\leq L}|\down_i(q)| \leq m\cdot 3^k=O(m)$.
\end{proof}

\subsection{Lemma~\ref{lemma:sum-q} Proof}

\sumq*

\begin{proof}
Consider points $r$ in $Q_i$ in an increasing order of $d(r, S_i\cup Q_i)$. There is a $1/2$ chance that the pivot $p_i$ is chosen such that $d(p_i, S_i\cup Q_i)$ is not larger than that of at least half of the points in $Q_i$, making them sparse and not in $Q_{i+1}$. Therefore $|Q_{i+1}|\leq |Q_i|/2$ in expectation. Given that $|Q_1| = O(m)$, we know $\sum_{1\leq i\leq L}E[|Q_i|] = O(m)$.
\end{proof}
\section{Parallel Deletion}\label{appendix:deletion}

\subsection{Algorithm}

The pseudocode for our batch deletion algorithm is shown in Algorithm~\ref{alg:grid-delete}. 
It takes as input $(S_i,S_i',p_i,q_i,d_i)$ and $H_i$ for $1\leq i\leq L$, and a batch of points $Q$ to be deleted.
We update the data structure level by level similar to the insertion algorithm, but in the opposite direction, starting at the last level $L$.
We define $Q_i$ for level $i$ as $Q\cap S_i$. From the property of the sparse partition, we have $Q_j\subseteq Q_i$ for all $i< j\leq L$. At each level, we delete each point in $Q_i$ from $S_i$, and also from $S_i'$ if it exists.

\begin{algorithm}[t]
\SetAlgoLined
  \SetKwInOut{Input}{Input}\SetKwInOut{Output}{Output}
  \Input{$(S_i, S_i', p_i, q_i, d_i)$ and $H_i$ for $1\leq i \leq L$; a batch $Q$ to be deleted.}

 \SetKwFunction{algo}{algo}\SetKwFunction{proc}{proc}
 \SetKwProg{myalg}{Algorithm}{}{}
 \SetKwProg{myproc}{Procedure}{}{}

 \myalg{\upshape\textsc{Main}()}{
   Determine $Q_i$ for all $1\leq i\leq L$. Specifically, for each level from $L$ to 1, compute $Q_i = Q\cap S_i$.\label{line:grid-delete:q-i}
 
   \textsc{Delete}($\emptyset$, $L$)\;\label{line:grid-delete:delete}

   Rebuild from the level with the smallest $i$ that needed a rebuild. Specifically, call \textsc{Build}$(S_i, i)$.\label{line:grid-delete:rebuild}
 }{}
   
 \myproc{\upshape\textsc{Delete}($\up_{i}$, $i$)\label{line:grid-delete:delete-procedure}}{
   $\up_{i-1}:=$ \textsc{GridDelete}($\up_{i}$, $i$)\;\label{line:grid-delete:grid}
   \textsc{HeapUpdate}$(i)$\;\label{line:grid-delete:heap}
   \lIf{$i-1\geq 1$}{\textsc{Delete}($\up_{i-1}$, $i-1$)}\label{line:grid-delete:delete-recur}
 }

 \myproc{\upshape\textsc{GridDelete}($\up_{i}$, $i$)}{
   Determine if $p_i$, $q_i$, or $d_i$ should change after deleting $Q_i$, which happens if at least one of $p_i$ or $q_i$ is in $Q_i$. If so, mark level $i$ for rebuild.\label{line:grid-delete:pivot}

   Insert each point in $\up_{i}$ into the dictionary of $S_i'$ in parallel.\label{line:grid-delete:insert}
   
   For each point $x$ in $Q_i$ in parallel, delete $x$ from $S_i$ and $S_i'$.\label{line:grid-delete:delete-q}

   For each point $r$ in $N_i(x, S_i)$ where $x\in Q_i$, check $N_{i-1}(r, S_{i-1})$. If $N_{i-1}(r, S_{i-1}) \subseteq Q_{i-1}$, then delete $r$ from $S_i$ and $S_i'$, and insert $r$ into the set $\up_{i-1}$.\label{line:grid-delete:up}

   \Return $\up_{i-1}$\;
 }

 \caption{Batch Delete}\label{alg:grid-delete}
\end{algorithm}

While the insertion algorithm moves sets of points $\down_{i}$ from level $i-1$ to level $i$, the deletion algorithm moves points in the opposite direction, from level $i+1$ to $i$.
We define $\up_{i}$ to be the set of points that move from level $i+1$ to level $i$,
i.e., $\up_{i} = \{x\in S_{i+1}: N_{i}(x,S_i)\subseteq Q_i\}$. 
They are the points $x$ in $S_{i+1}$ that only contain points from $Q_i$ in their neighborhoods $N_i(x, S_i)$ in level $i$; when $Q_i$ is deleted, they will become sparse in $S_i$, and will no longer be in $S_{i+1}$. Eventually, the points in $\up_i\setminus \up_{i-1}$ are added to both $S_i$ and $S_i'$.

Initially, we determine $Q_i$ for all levels via a backward pass starting from level $L$ (Line~\ref{line:grid-delete:q-i}).
Given $Q_i\subseteq Q_j$ for $i>j$, when a point is added to $Q_i$, it will be added to all $Q_j$ where $j<i$.
We pass an empty $\up_{L}$ to procedure \textsc{Delete} (Line~\ref{line:grid-delete:delete}).
In the procedure \textsc{Delete} (Line~\ref{line:grid-delete:delete-procedure}), the algorithm performs the deletion from the grid at level $i$ (Line~\ref{line:grid-delete:grid}), updates the heap (Line~\ref{line:grid-delete:heap}), and then recursively calls \textsc{Delete} on level $i-1$ until deletion is complete on level $1$ (Line~\ref{line:grid-delete:delete-recur}). 
Like in the insertion algorithm, we determine whether to rebuild at each level, but unlike insertion we delay the rebuild until the end of the algorithm (Line~\ref{line:grid-delete:rebuild}).
We call rebuild just once, on the level with the smallest $i$ that needs a rebuild (as this will also rebuild all levels greater than $i$).

In the procedure \textsc{GridDelete}($up_i,i$), we determine if the pivot needs to change based on whether at least one of $p_i$ and $q_i$ are in $Q_i$. If so, we mark level $i$ for rebuilding (Line~\ref{line:grid-delete:pivot}).
We then insert $\up_{i}$ into $S_i'$ and delete the points in $Q_i$ from $S_i$ and $S_i'$ if they exist (Lines~\ref{line:grid-delete:insert}--\ref{line:grid-delete:delete-q}).

We determine $\up_{i-1}$ by finding the points that will become sparse in level $i-1$ (Line~\ref{line:grid-delete:up}).
Since the movement of $\up_{i-1}$ from level $i$ to $i-1$ is due to the deletion of $Q_i$, we enumerate the candidates for $\up_{i-1}$ from $N_{i}(x,S_{i})$ where $x\in Q_i$. Then, for each candidate $r$, we check if $N_{i-1}(r, S_{i-1})$ only consists of points in $Q_{i-1}$, which are to be deleted in $i-1$. If so, $r$ will move up to a level less than or equal to $i-1$, and so we add $r$ to $\up_{i-1}$. 
A few details need to be noted to make the computation of $\up_{i-1}$ take $O(m)$ work. First, when checking the neighborhood $N_i(x,S_i)$ for the candidates $r$, we should only check a neighboring box if it contains at most one point, since otherwise the candidate would not be sparse in $S_{i-1}$. This bounds the work of enumerating candidates to $O(3^k\cdot m)$.
We next describe the process for checking for each candidate $r$ whether $N_{i-1}(r, S_{i-1})$ contains only points in $Q_{i-1}$.
Naively checking all of the points in the neighborhoods of each candidate could lead to quadratic work.
In our algorithm, we use a parallel dictionary to implement a temporary, empty grid structure with grid size equal to that of $S_{i-1}$. We insert points in $Q_{i-1}$ into this grid in parallel. This takes $O(|Q_i|)$ work and $O(\log^* m)$ depth per level.
Now for each candidate $r$, we compare the number of points in each box of $N_{i-1}(r, S_{i-1})$ to the corresponding box in the temporary grid, and if they are equal, then we know that the box only contains points in $Q_{i-1}$.

\subsection{Analysis}

Since the probability of a rebuild at each level $i$ is $|Q\cap S_i|/|S_i|$ and the work for the rebuild is $O(|S_i|)$, the expected work of rebuilding at level $i$ is $O(|Q|) = O(m)$.
Since we do at most one rebuild across all levels, it contributes $O(m)$ in expectation to the work and $O(\log (n+m) \log^*(n+m))$ \whp{} to the depth.

The total size of $Q_i$ and $up_i$ across all of the levels is proportional to the batch size. Since the point movement is the exact opposite of that of batch insertion, the proof is very similar. We omit the proof and just present the lemmas below.
\begin{lemma}\label{lemma:packing-up}
  $|\bigcup_{1\leq i\leq L} \up_i| \leq m\cdot 3^k = O(m)$
\end{lemma}

\begin{lemma}\label{lemma:sum-q-delete}
  $\sum_{1\leq i\leq L} E[|Q_i|] = O(m)$
\end{lemma}

For the rest of the algorithm, not including the rebuild and heap update, it follows from Lemmas \ref{lemma:packing-up} and \ref{lemma:sum-q-delete} and a similar analysis to the insertion algorithm that Lines~\ref{line:grid-delete:insert}--\ref{line:grid-delete:up} take $O(m)$ amortized work in expectation and $O(\log (n+m)\log^* (n+m))$ depth across all levels. 
Therefore, we can maintain a sparse partition under a batch of $m$ deletions in $O(m)$ amortized work in expectation and $O(\log (n+m) \log^*(n+m))$ depth \whp{}, as stated in Lemma~\ref{lemma:grid-update-delete}.
We describe the cost of the heap update in Appendix~\ref{appendix:heap-update}.

\section{Maintaining the Heaps $H_i$} \label{appendix:heap-update} 
On each level $i$, we maintain a parallel min-heap $H_i$ storing the restricted distances for each point in $S_i'$. In this section, we elaborate on the \textsc{HeapUpdate} procedure on Line~\ref{line:grid-insert-short:heap} of Algorithm~\ref{alg:grid-insert-short} and Line~\ref{line:grid-delete:heap} of Algorithm~\ref{alg:grid-delete}. Each call to \textsc{HeapUpdate}$(i)$ updates $H_{i+l}$ for $0\leq l \leq k$ so that they contain the updated restricted distances in level $i$. 

\subsection{Definitions}
Here we define some terms that we use in the algorithm description and analysis. 

During a batch insertion, we process each level $i$ with inputs $Q_{i}$ and $\down_{i}$ (Algorithm~\ref{alg:grid-insert-short}).
By definition, $\down_{i}$ contains the points moved from level $i-1$ to levels $i$ and greater.
We say that point $x$ \defn{starts moving} at level $i$ if $x\in \down_{i+1} \setminus \down_{i}$. We say that point $x$ \defn{stops moving} at level $i$ if $x\in \down_{i}\setminus \down_{i+1}$, or if point $x\in Q_{i}\setminus Q_{i+1}$, i.e., $x$ is sparse and stays in $S_i'$.
Finally, point $x$ \defn{moves through} level $i$ if it is in $\down_{i}\cap \down_{i+1}$. 


During a batch deletion, we process each level $i$ with input $\up_{i}$, and delete points in $Q$ from the level if they exist (Algorithm~\ref{alg:grid-delete}).
Similar to insertion, $up_{i}$ contains the points moved from level $i+1$ to levels $i$ and less.
We say that point $x$ starts moving at level $i$ if $x\in up_{i-1}\setminus up_{i}$; or if $x\in Q$ is deleted from $S_i'$.
We say that point $x$ stops moving at level $i$ if $x\in up_{i}\setminus up_{i-1}$.
Finally, we say that point $x$ moves through level $i$ if it is in $up_{i}\cap up_{i-1}$.


\begin{algorithm}[t]
\SetAlgoLined
  \SetKwInOut{Input}{Input}
  \SetKwInOut{Output}{Output}
  \Input{$(S_i,S_i',p_i,q_i,d_i)$ and $H_i$ with updated grids; point set $M_1$ that start moving at level $i$ and point set $M_2$ that stop moving at level $i$.
  }

  \SetKwFunction{algo}{algo}\SetKwFunction{proc}{proc}
  \SetKwProg{myalg}{Algorithm}{}{}
  \SetKwProg{myproc}{Procedure}{}{}

  \myalg{\upshape\textsc{HeapUpdate-Naive}($i$)}{\label{line:heap-naive:proc}
    Batch delete $d_i^*(p)\ \forall p \in M_1$ from $H_i$.\label{line:heap-naive:del}
    
    \For{$0\leq l\leq k$}{\label{line:heap-naive:for1-start}
      Batch delete $d_{i+l}^*(q)$ from $H_{i+l}$ such that $d^*_{i+l}(q)=d(q,p)$ for some $p \in  M_1$.\label{line:heap-naive:for1-del}
      
      In parallel, recompute $d_{i+l}^*(q)$ using the grid of ${S}_{i+l}$ for all $q$ whose old $d_{i+l}^*(q)$ was just deleted.\label{line:heap-naive:for1-recomp}
      
      Batch insert new $d_{i+l}^*(q)$ for all $q$ into $H_{i+l}$.\label{line:heap-naive:for1-insert-end}
    }
    
    Compute and batch insert $d_i^*(p)\ \forall p\in M_2$ into $H_i$.\label{line:heap-naive:insert} 
    
    \For{$0\leq l\leq k$}{\label{line:heap-naive:for2-start}
      Batch delete from $H_i$ the $d^*_{i+l}(q)$ for each point $q\in {S}_{i+l}'$, if $d(q,p)< d^*_{i+l}(q)$ for some $p\in M_2$.\label{line:heap-naive:for2-del}
      
      Batch insert into $H_i$ the new $d^*_{i+l}(q):=d(q,p)$ for each aforementioned point $q$ on the previous line.\label{line:heap-naive:for2-insert-end}
    }
  }
  \caption{Naive Heap Update 
  }\label{alg:heap-update-naive}
\end{algorithm}

\subsection{Parallel Update} 
The heap $H_i$ contains the restricted distance $d_i^*(q)$ for $q\in S_i'$. By definition, $d_i^*(q)$ is the closest distance of $q$ to another point in $S_{i-l}'$ where $0\leq l\leq k$ ($k$ is the dimension of the data set). Therefore, following an update on $S_i'$, we need to update the $d_i^*(q)$ in $H_{i+l}$ for $0\leq l\leq k$, and $q\in S_{i+l}'$.
Specifically, the update happens when $d_i^*(q)=d(q,p)$, but $p$ starts moving at level $i$; or when $p$ stops moving at level $i$ and $d(q,p)<d_i^*(q)$.
Since the update of $S_i'$ initiates the update on some heap $H_{i+l}$ for $0\leq l\leq k$, we call level $i$ the \defn{initiator} and each 
heap $H_{i+l}$ a \defn{receptor} of the initiator.

We first start with a more intuitive but less parallel algorithm, which is shown in Algorithm~\ref{alg:heap-update-naive}. 
It takes as input the updated sparse partition $(S_i,S_i',p_i,q_i,d_i)$, the set of points $M_1$ that start moving at level $i$, and a set of points $M_2$ that stop moving at level $i$. 
Lines \ref{line:heap-naive:del}--\ref{line:heap-naive:for1-insert-end} process the set of points $M_1$. We first batch delete $d_i^*(p)$ from $H_i$ for all $p$ in $M_1$, since they start moving at level $i$ (Line \ref{line:heap-naive:del}). 
We then update each receptor heap if it stores some $d_i^*(q)$ that is generated by a deleted point $p\in M_1$. To know each potential point $q$, we iterate over the neighborhood of each $p$ in $M_1$ and then check if $d_i^*(q)$ needs to be updated (Lines \ref{line:heap-naive:for1-start}--\ref{line:heap-naive:for1-insert-end}). 
Lines \ref{line:heap-naive:insert}--\ref{line:heap-naive:for2-insert-end} process $M_2$. We compute new restricted distances and batch insert the points in $M_2$ into $H_i$, since they stop moving at level $i$ (Line \ref{line:heap-naive:insert}). 
We then update the receptor heaps when a heap contains the restricted distance of point $q$, but $q$ has a smaller distance to a newly inserted $p\in M_2$ than to its previous closest point (Lines \ref{line:heap-naive:for2-start}--\ref{line:heap-naive:for2-insert-end}).

\begin{figure}[t]
\centering
\begin{minipage}[t]{0.49\textwidth}
  \includegraphics[width=\textwidth]{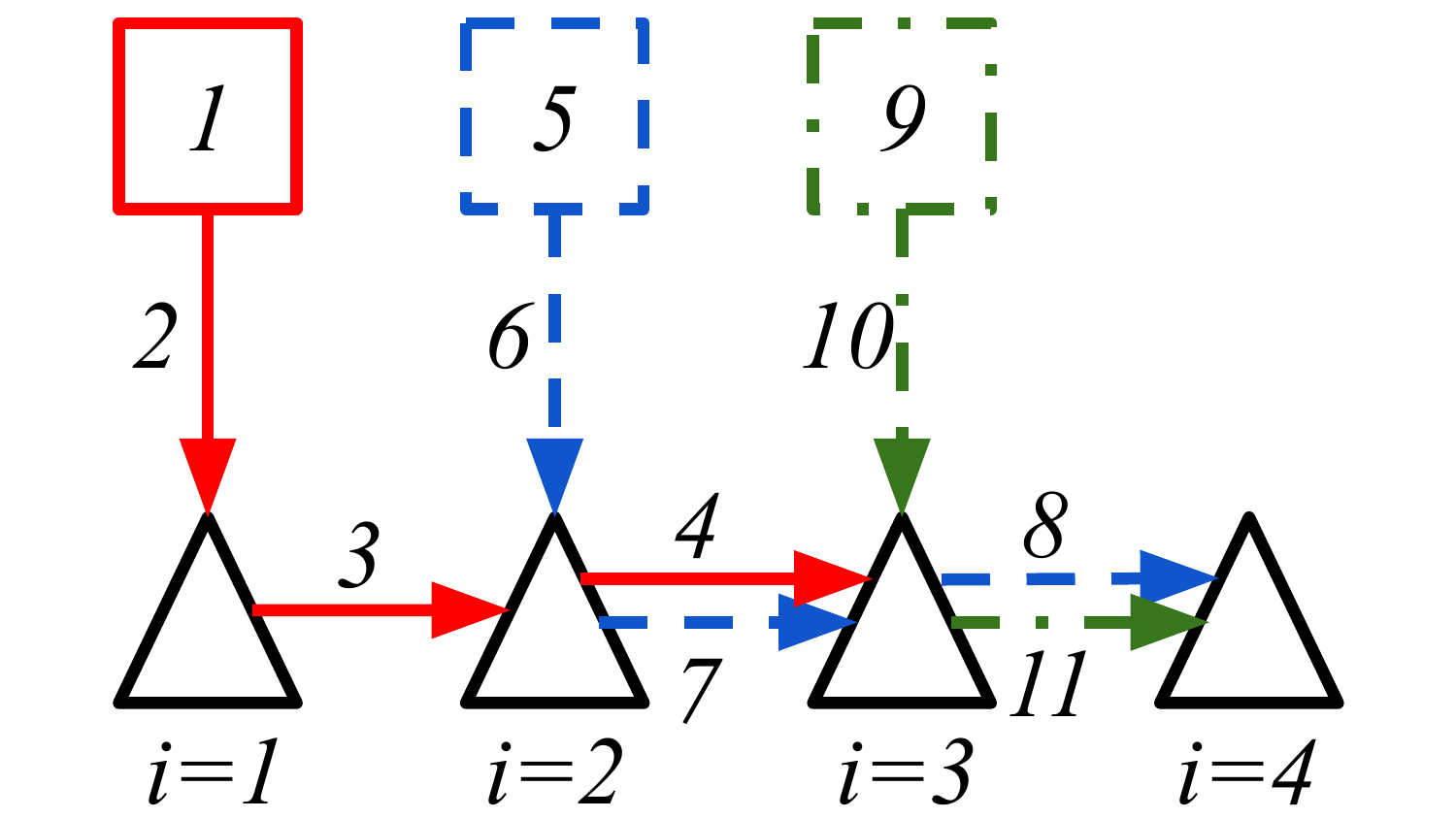}
  \subcaption{Naive insertion.}\label{fig:heap-insert-naive}
\end{minipage} %
\begin{minipage}[t]{0.49\textwidth}
  \includegraphics[width=\textwidth]{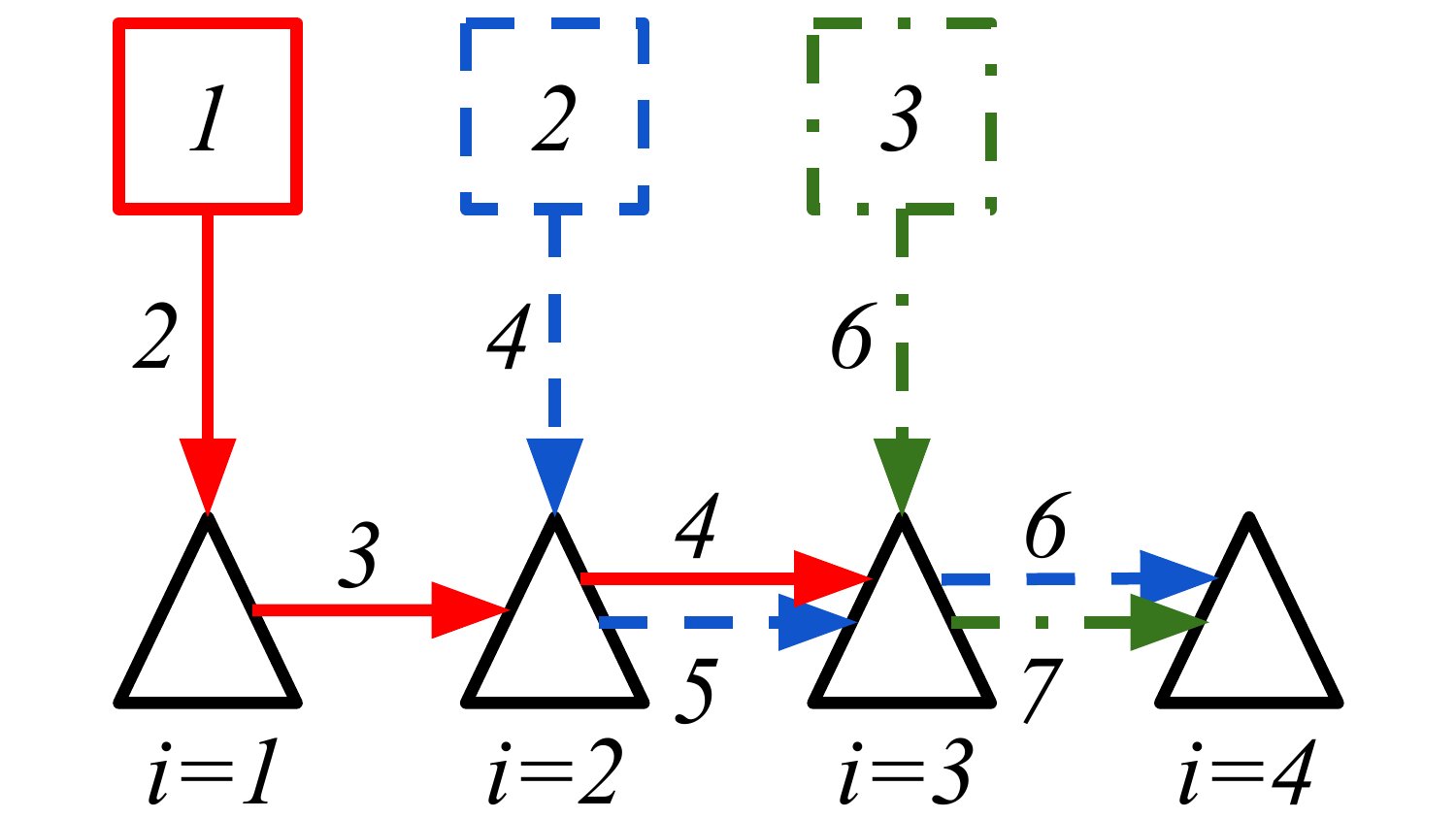}
  \subcaption{Parallel insertion.}\label{fig:heap-insert-parallel}
\end{minipage} %
\begin{minipage}[t]{0.49\textwidth}
  \includegraphics[width=\textwidth]{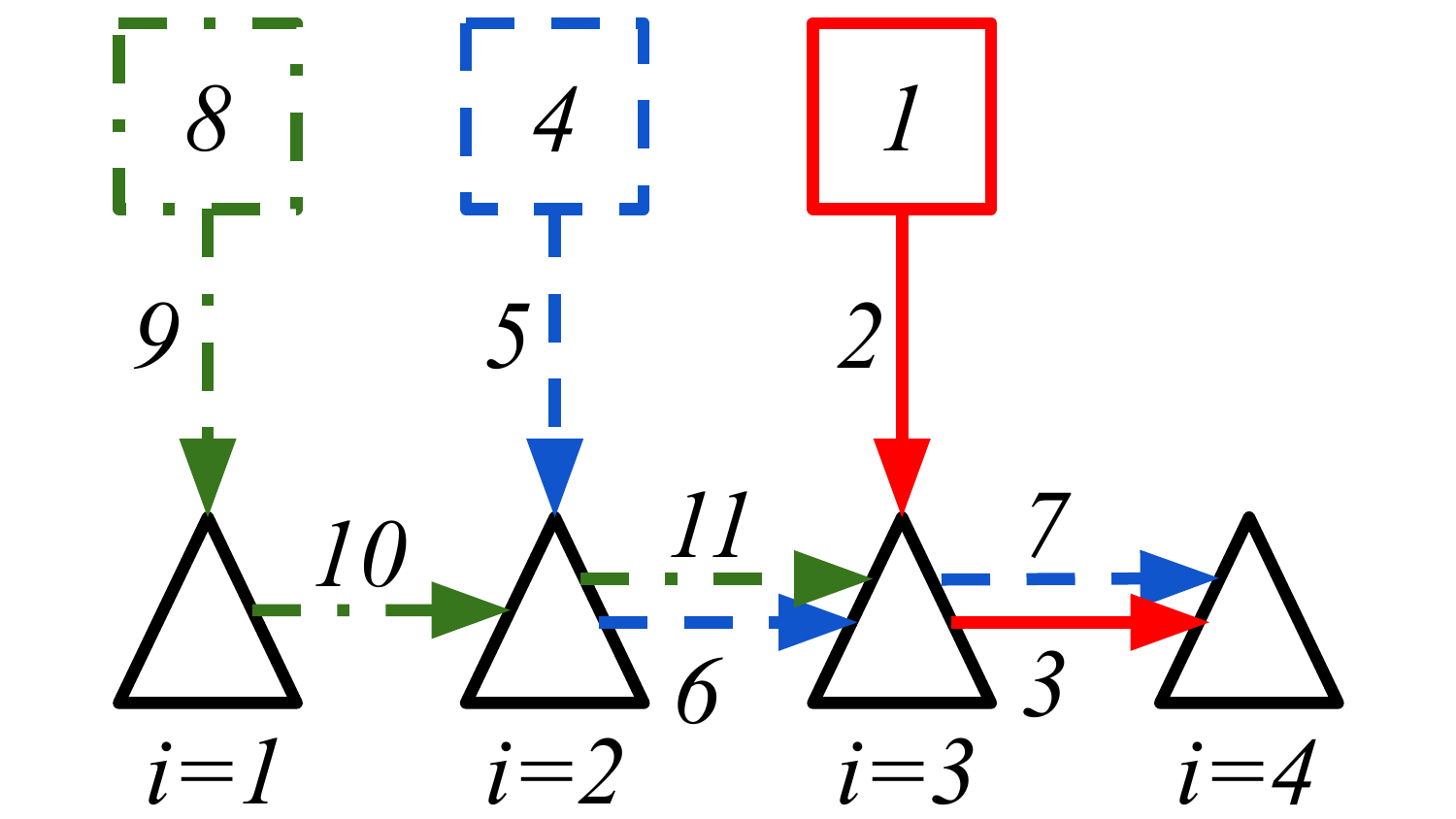}
  \subcaption{Naive deletion.}\label{fig:heap-delete-naive}
\end{minipage} %
\begin{minipage}[t]{0.49\textwidth}
  \includegraphics[width=\textwidth]{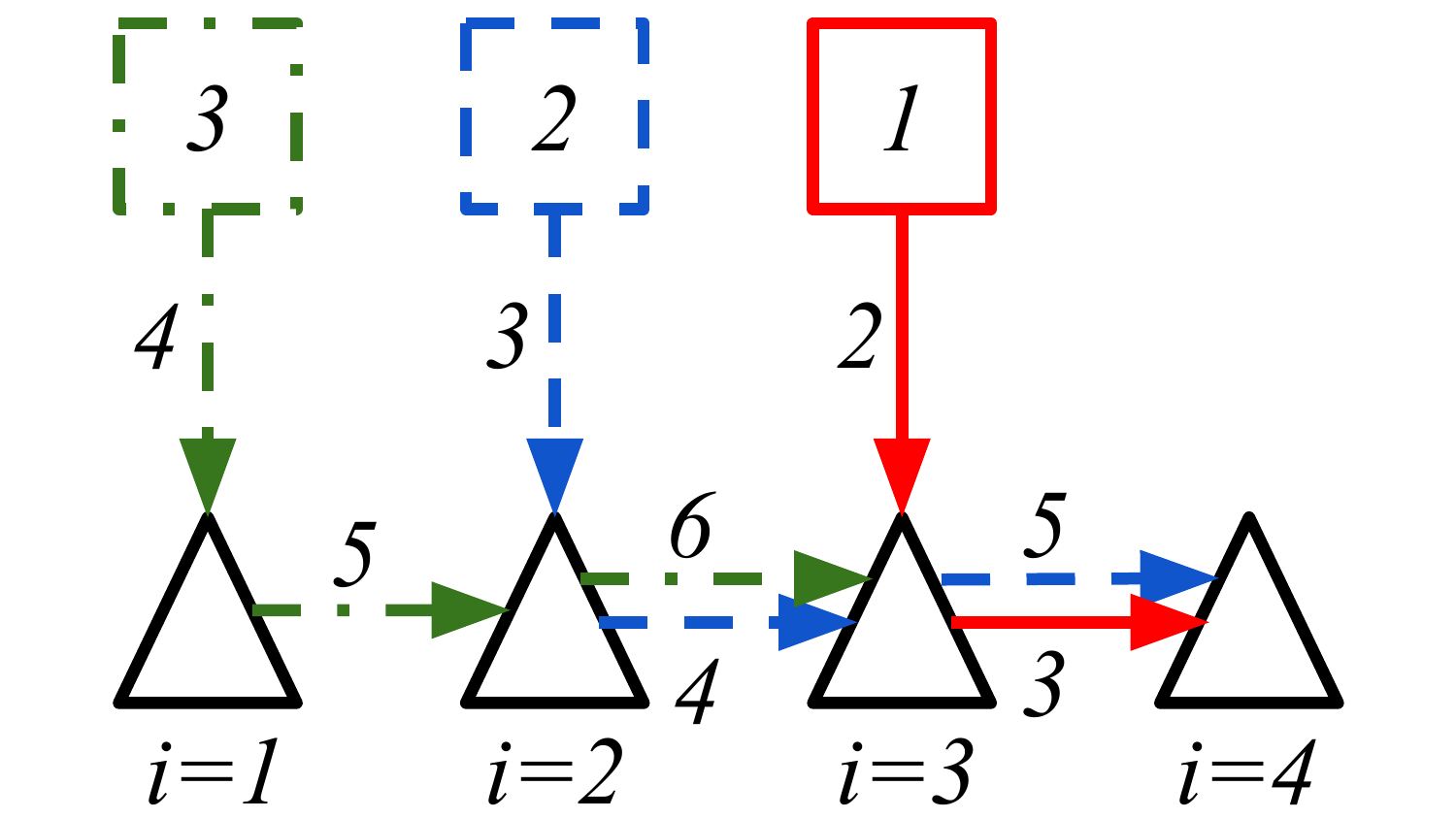}
  \subcaption{Parallel deletion.}\label{fig:heap-delete-parallel}
\end{minipage} %
\caption{This figure shows examples of how heap updates work during insertion and deletion. Calls to \textsc{GridInsert} and \textsc{GridDelete} are shown by the boxes. Calls to \textsc{HeapUpdate} are shown by the arrows, whereas the actual heaps $H_i$ are shown by the triangles. The example considers dimension $k=2$ and shows $H_i$ for $i=1,2,3,4$, but considers only updating levels $i=1,2,3$. We number the calls by the order that they happen, and two calls have the same number if they can be done in parallel. For clarity, we annotate the operations associated with different levels in different colors and line-styles.}\label{fig:heap-update}
\end{figure}

Using our batch-parallel binary heap, which we will describe in detail in Appendix~\ref{appendix:heap}, each batch update of the heap takes $O(\log (n+m))$ depth. The computation of the new restricted distances takes $O(1)$ depth. Therefore, the naive heap update algorithm takes $O(\log (n+m))$ depth per call. 
For the batch insertion algorithm in Algorithm~\ref{alg:grid-insert-short},
\textsc{GridInsert} on level $i+1$ is blocked by \textsc{HeapUpdate-Naive} of level $i$, to prevent multiple initiators updating the same receptor simultaneously.
Figure~\ref{fig:heap-insert-naive} illustrates four levels of the data structure. The updates of level 2 shown in blue dashed lines are blocked until the completion of level 1 shown in red solid lines, and similarly for the remaining levels.
Since there are $O(\log (n+m))$ levels \whp{}, this leads to a overall depth of $O(\log^2 (n+m))$ \whp{} for the heap updates. A similar argument applies for the batch deletion algorithm,
except that the order that the levels are updated proceeds in descending value of $i$ starting with $L$, as shown in Figure~\ref{fig:heap-delete-naive}.

We improve the depth of the heap update to $O(\log (n+m))$ \whp{} by pipelining the heap updates. For insertion, a crucial observation is that the grid update of level $i+1$ can start right after that of $i$. Meanwhile, we run the heap updates in lock-step in parallel with the grid updates, as illustrated in Figure~\ref{fig:heap-insert-parallel}---for each heap, the update from level $i+1$ can start right after that of level $i$ is completed.
For example, the update on $H_2$ initiated by level 2 (blue dashed arrow with a value of 4) can start right after the update initiated by level 1 (red arrow with a value of 3).
Like insertion, deletion follows similar strategy in the reverse order, as illustrated in Figure~\ref{fig:heap-delete-parallel}.

\subsection{Analysis (Theorem~\ref{thm:overall})}\label{appendix:overall-analysis}

All updates on the heaps in our data structure are a result of points that start or stop moving at some level. First, we are concerned with those added to or deleted from $S_i'$, and hence $H_i$.
Since the $S_i'$ for $1\leq i\leq L$ are disjoint sets, $O(m)$ points from $Q$ are inserted or deleted from $H_i$ across all $1\leq i\leq L$.
Second, points in $\down_i$ and $\up_i$ for $1\leq i\leq L$ also cause heap updates, and the total number of heap updates from these points is $O(m)$ by Lemmas~\ref{lemma:packing} and~\ref{lemma:packing-up}. 
Therefore, across all levels, there are $O(m)$ updates to the heap.
For all $p$ that start or stop moving across all levels, finding and recomputing $d_{i+l}^*(q)$ for $0\leq l\leq k$ takes $O(m)$ work and $O(1)$ depth. This is because the cost associated with each $q$ whose $d^*(q)$ depends on $p$ involves searching the constant number of boxes surrounding $q$ and $p$ in $k$ sparse sets, where $k$ is constant, as shown by Golin et al.~\cite{golin1998randomized}.
In \cref{appendix:heap}, we will show that the total work for a batch of $r$ updates to our parallel heap is $O(r(1+\log ((n+r)/r)))$.
Therefore, the total work for heap operations is hence $O(m(1+\log ((n+m)/m)))$.

In \cref{appendix:heap}, we will show that batch insertions and deletions on a heap take $O(\log (n+m))$ depth.
Since the heap updates are performed in parallel with the grid updates, they 
are not on the critical path of the computation. The total depth is dominated by the grid updates, which is $O(\log (n+m)\log^*(n+m))$ \whp{}.

This leads to our overall bounds of amortized $O(m(1+\log ((n+m)/m)))$ work and $O(\log (n+m) \log^*(n+m))$ depth \whp{} for an update of size $m$, where the amortization comes from resizing the parallel dictionaries.
\section{Parallel Batch-Dynamic Binary Heap}\label{appendix:heap}

Our parallel heap that supports batch updates (inserts and deletes) and finding the minimum element (find-min) efficiently is crucial to our data structure for the closest pair.
Although we could implement a parallel heap using a parallel binary search tree, which supports a batch of $m$ updates to a set of $n$ elements in $O(m\log (n+m))$ work and $O(\log (n+m))$ depth~\cite{CLRS}, it supports more functionality (i.e., returning the minimum $K$ elements) than we need. In fact, the $O(m\log (n+m))$ work bound is tight for a binary search tree, since we can use it for comparison sorting.

For our batch-dynamic heap introduced in Section~\ref{sec:heap}, it only needs to support the find-min operation rather than maintaining the full ordering, and hence it has a better work bound of $O(m(1+\log((n+m)/m)))$. 
Furthermore, it allows us to construct the initial heap in linear work (by setting $m$ to the number of points and $n=0$ in the work bound), as needed for Theorem~\ref{theorem:build}. The pseudocode of our algorithm is shown in \cref{alg:heap}, and its discussion is in Section~\ref{sec:heap}.  We present the analysis in this section.

\begin{algorithm}[t]
  \caption{Parallel \heapify algorithm}\label{alg:heap}
  \SetAlgoLined
  \SetKwInOut{Input}{Input}
  \SetKwInOut{Output}{Output}
  \Input{A binary min-heap of size $n$ with $m$ updates, each of which is a triple $(v_i, k_i, k_i')$, indicating to update key $k_i$ to $k'_i$ on node $v_i$.}
  \Output{An updated binary heap.} 

  \SetKwFunction{algo}{algo}
  \SetKwFunction{proc}{proc}
  \SetKwProg{myalg}{Algorithm}{}{}
  \SetKwProg{myproc}{Procedure}{}{}

  Let $S^+$ be the set of nodes with keys to be increased.\label{line:down-start}
  
  Use integer sort to group the nodes in $S^+$ to $S^+_l$ by the level $l$ in the heap (the root has level 0).\label{line:heap-sort-1}

  \For {\upshape $l\gets \lfloor \log_2 n\rfloor-1$ to $0$} {
      \For {$v_i\in S^+_l$ \textbf{in parallel}} {
        \textsc{Down-Heap}($v_i$)\label{line:down-end}
      }
  }
  Let $S^-$ be the set of nodes with keys to be decreased.\label{line:up-start}
  
  Use integer sort to group the nodes in $S^-$ to $S^-_l$ by the level $l$ in the heap.\label{line:heap-sort-2}
  
  \For {\upshape $l\gets 1$ to $\lfloor \log_2 n\rfloor$} {
      \For {$v_i\in S^-_l$ \textbf{in parallel}} {
        \textsc{Up-Heap}($v_i$)\label{line:up-end}
      }
  }
\end{algorithm}

\subsection{Analysis of the \heapify Algorithm}

\myparagraph{Correctness}
The correctness can be shown inductively on subtrees of increasing height.
For the base case, all leaf nodes are valid binary heap subtrees, each containing one node.
Then on the first iteration, we run \textsc{Down-Heap} for updated keys on the second to last level.
If the increased keys violate the heap property, then \textsc{Down-Heap} will heapify this subtree, which has two levels.
Similarly, for each node $v$ with increased keys on level $i$, both of $v$'s childrens' subtrees are valid binary heap subtrees, and so after \textsc{Down-Heap}, the subtree rooted at $v$ is a valid binary heap subtree.
The correctness for \textsc{Up-Heap} can be shown symmetrically. 
The main difference is that in \textsc{Up-Heap}, the update paths can overlap, but the correctness is guaranteed since it is implemented in a round-synchronous manner.
In addition, when both children are updating a parent, we only allow the smaller child to continue its update path, while terminating that of the larger child, thereby satisfying the heap property.

\myparagraph{Work}
We now consider the work of this algorithm.
Let $h$ be the height of the binary heap.
For the worst case analysis, we always assume that \textsc{Down-Heap} pushes a node to the leaf and that \textsc{Up-Heap} pushes a node to the root.
The case for \textsc{Down-Heap} is simple---for $m=2^r-1$ increase-keys, the worst case is when they are in the top $r$ levels.
Each \textsc{Down-Heap} is independent and the total work is 
$$\sum_{i=0}^{r}{2^i(h-i)}=O(m(h-r+1))=O\left(m\left(1+\log\left(\frac{n}{m}\right)\right)\right).$$ 

The work for \textsc{Up-Heap} is more involved. 
Let $m_i$ be the number of increase-keys on level $i$. We know that $m_i\le 2^i$ and $\sum{m_i}\le m$. 
For level $i$, the work for all calls to \textsc{Up-Heap} is upper bounded by the number of nodes on the path from the root to all updated nodes in level $i$.
It can be shown that the number of such nodes is $O\left(m_i\left(1+\log \frac{n}{2^{h-i}m_i}\right)\right)$ (Theorem 6 in~\cite{blelloch2016just}).
Hence, the overall work for all levels is $W=O\left(\sum_{i=0}^{\log_2 n}{m_i\left(1+\log \frac{n}{2^{h-i}m_i}\right)}\right)$.
Let $m'=\sum_i m_i$, and we know that $m'\le m$.
To bound the  work, we consider the maximum value of $W$ for any given $m'$.
We can use the method of Lagrange multipliers, 
and compute the partial derivative of $m_i$ (without the big-$O$), which solves to

\begin{align*}
\frac{\partial}{\partial m_i}W
=\frac{\partial}{\partial m_i}\left(m_i\left(1-\log_2 \frac{2^{h-i}m_i}{n}\right)\right)
=\log_2 \frac{n}{2^{h-i}m_i}-\frac{1}{\ln 2}+1.
\end{align*}

Since the constraint for $\sum{m_i}$ is linear, $W$ is maximized when $\frac{\partial}{\partial m_i}W=\frac{\partial}{\partial m_j}W$ for all levels $0\le i,j\le h$, which solves to $m_i=cm'/2^{h-i+1}$ for some value of $c$ such that $m'=\sum m_i$. Note that $1<c<2$. 
Plugging this in gives
\begin{align*}
W&=O\left(\sum_{i=0}^{\log_2 n}\frac{m'}{2^{h-i+1}}\left(1+\log \frac{n}{2^{h-i} (cm'/2^{h-i+1})}\right)\right)\\
&=O\left(\left(\sum_{i=0}^{\log_2 n}\frac{m'}{2^{h-i+1}}\right)\left(1+\log \left(\frac{n}{m'}\right)\right)\right)=O\left(m\left(1+\log \left(\frac{n}{m}\right)\right)\right).
\end{align*}

In addition to \textsc{Down-Heap} and \textsc{Up-Heap}, we also need to integer sort the updates on Lines \ref{line:heap-sort-1} and~\ref{line:heap-sort-2}, which takes $O(m)$ work.
Hence, the total work for \cref{alg:heap} is $O\left(m\left(1+\log \left(\frac{n}{m}\right)\right)\right)$. 

\myparagraph{Depth}
The integer sort on Lines \ref{line:heap-sort-1} and~\ref{line:heap-sort-2} takes $O(\log m)$ depth.
Directly running \cref{alg:heap} gives $O(\log^2 n)$ depth---there are $O(\log n)$ tree levels, and on each level, \textsc{Up-Heap} or \textsc{Down-Heap} requires $O(\log n)$ depth.
We can improve the depth bound to $O(\log n)$ using pipelining, as discussed in \cref{sec:heap}.
The \textsc{Up-Heap} and \textsc{Down-Heap} calls at all levels will have begun by the $i$'th round, and each call takes $O(\log n)$ rounds to finish. Each round takes $O(1)$ depth, and so the overall depth is $O(\log n)$. The pipelining does not increase the work.

Combining the work and depth analysis gives the bounds in \cref{thm:heap}.


\section{Correctness of the Simplified Data Structure}\label{sec:appendix-simplified}
Our implementation of the parallel batch-dynamic data structure maintains just a single heap $H^*$ for the last $\lceil \log_3 {2\sqrt{k}} \rceil$ levels. 
For correctness, we prove that any point pair $(a,b)$ where $a,b\in S\setminus S_{j}, \forall j\leq L-\lceil \log_3 {2\sqrt{k}} \rceil$ cannot give rise to the closest pair distance, which we denote as $\delta(S)$. 

\begin{lemma}\label{lemma:simplify}
$\delta(S) < d(a,b)$ for any $a,b\in S\setminus S_{j}, \forall j\leq L-\lceil \log_3 {2\sqrt{k}} \rceil$.
\end{lemma}

\begin{figure}[t]
\center
    \includegraphics[width=0.2\textwidth]{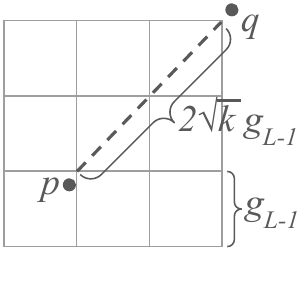}
  \caption{Illustration for the proof of Lemma~\ref{lemma:simplify}: Let $\delta(S)$ be $p$ and $q$. Suppose by contradiction, $d(p,q) > 2\sqrt{k}g_{L-1}$, where $k=2$ in this example. Then $p$ and $q$ will be sparse in level $L-1$, hence all the points will be sparse in level $L-1$.}\label{fig:heapproof}
\vspace{1em}
\end{figure}

\begin{proof}
The side length of the grid $G_i$ at level $i$ is $g_i=d_i/6k$, as defined earlier.
Without loss of generality, consider level $j$, and let $a$ and $b$ be two points such that $a,b\in S_{j-1} \setminus S_j$ (this can easily be generalized to $a$ and $b$ being sparse on different levels that are both $< j$).
We know
$d(a,b)>g_{j-1}\geq 3 g_{j}$ by properties of the sparse partition (Section~\ref{sec:golin:sp}).
On the other hand, we also know that $\delta(S) \leq 2\sqrt{k} g_{L-1}$
, since otherwise, the pair that defines $\delta(S)$ would have been sparse in level $L-1$, and the last level $L$ would not have existed, which is a contradiction (see Figure~\ref{fig:heapproof}).
Given the property of the sparse partition,
we have $3 g_{j} \geq 2\sqrt{k} g_{L-1}$ for all $j\leq L-\lceil \log_3 {2\sqrt{k}} \rceil$. We can verify that this is true as $3 \cdot g_{L-\lceil \log_3 {2\sqrt{k}} \rceil} \geq 3 \cdot 3^{\lceil \log_3 {2\sqrt{k}} \rceil} \cdot g_L \geq 3^{\lceil \log_3 {2\sqrt{k}} \rceil} \cdot g_{L-1} \geq 2\sqrt{k} \cdot g_{L-1}$.
Therefore, $d(a,b)>\delta(S)$.
\end{proof}

Since $a$ and $b$ satisfying Lemma~\ref{lemma:simplify} are both sparse in some $S_h$ where $h<L-\lceil \log_3 {2\sqrt{k}} \rceil$, it follows that $d(a,q)>d(q,S_{L-\lceil \log_3 {2\sqrt{k}} \rceil})$ and $d(b,q)>d(q,S_{L-\lceil \log_3 {2\sqrt{k}} \rceil})$ for any $q\in S_{L-\lceil \log_3 {2\sqrt{k}} \rceil}$.
Therefore, the closest pair distance $\delta(S) = d(p,q)$ for some $p,q\in S_{L-\lceil \log_3 {2\sqrt{k}} \rceil}$, and will not involve $a$ or $b$.

\section{Static Algorithms and Implementations}\label{appendix:static}

In addition to our batch-dynamic closest pair algorithm, 
we implement several parallel algorithms for the static closest pair problem, which we describe in this section. We evaluate all of them against each other, and compare them to our parallel batch-dynamic algorithm in Section~\ref{sec:experiment}. As far as we know, this paper presents the first experimental study of parallel algorithms for the static closest pair problem.

\subsection{Divide-and-Conquer Algorithm}
The first divide-and-conquer algorithm for closest-pair was introduced by Bentley~\cite{bentley1975multidimensional},
and has $O(n\log n)$ work and is optimal in the algebraic decision tree model. Blelloch and Maggs~\cite{blelloch2010parallel} parallelize this algorithm, and their algorithm takes $O(n\log n)$ work and $O(\log^2 n)$ depth. Atallah and Goodrich~\cite{atallah1986efficient} present another parallel algorithm based on multi-way divide-and-conquer, which takes $O(n\log n\log\log n)$ work and $O(\log n\log\log n)$ depth. 

We implement the divide-and-conquer algorithm by Blelloch and Maggs~\cite{blelloch2010parallel}. The main idea of the algorithm is to divide the space containing all the points $S$ along an axis-aligned hyperplane by the median point along a dimension fixed throughout the algorithm, to form left and right subproblems. We then recursively find the closest pair in each of the two subproblems in parallel to obtain results $\delta_L$ and $\delta_R$. Then, we merge the two subproblems, and consider the points near the median point, which are the points within a distance of $\min\{\delta_L, \delta_R\}$ from the median point. We call the set of such points a \defn{central slab}, and use an efficient ``boundary merging'' technique to obtain $\delta_M$.
The closest pair will have distance $\delta(S)=\min\{\delta_L, \delta_R, \delta_M\}$. Finding the median and performing the merge can be done using standard parallel primitives.

Blelloch and Maggs'~\cite{blelloch2010parallel} parallel algorithm requires the central slab to be sorted in a dimension $d$ different from the dimension that is used to divide the problem. Their algorithm orders the points along $d$ by performing recursive partitioning and merging at each level of the divide-and-conquer algorithm. Since the central slab can be linear in size, the merging algorithm is efficient in theory.
However, we find that the central slab is very small for inputs that arise in practice. Therefore, in our algorithm, we simply use comparison sorting to sort the central slab when needed
without using partitioning and merging, which results in better performance in practice. We also coarsen the base case, and switch to a quadratic-work brute-force algorithm when the subproblem size is sufficiently small.

\subsection{Rabin's Algorithm}
Rabin's algorithm~\cite{rabin1976probabilistic} is the first randomized sequential algorithm for the problem. Assuming a unit-cost floor function, Rabin's algorithm has $O(n)$ expected work.
MacKenzie and Stout~\cite{mackenzie1998ultrafast} design a parallel algorithm based on Rabin's algorithm, and achieve $O(n)$ work and $O(1)$ depth in expectation.
Specifically, the algorithm first takes a random sample of $n^{0.9}$ points, and finds the closest pair on this sample recursively. Then it forms a grid structure with side length equal to the closest pair distance. Each box along with its points are classified into sparse or dense based its point count.
The pairwise distances of the dense points (fewer than $n^{0.4}$) are computed using a quadratic work algorithm, after which the minimum is taken. Then, all points find their closest sparse points by checking neighboring sparse boxes in parallel, after which the minimum of these distances are taken. The algorithm uses some parallel primitives with high constant factor overheads, and are unlikely to be practical. 

We design a simpler parallel version of Rabin's algorithm.
Our algorithm takes a sample of $n^c$ points where $c<1$, and recursively computes the closest distance $\delta'$ of the sample. Then, we construct a grid structure on all of the points $S$ using a parallel dictionary, where the box size is set to $\delta'$. For each point $x\in S$, we find its closest point by exploring $N(x, S)$, and then take the minimum among that of all $x$ to obtain $\delta(S)$.
In terms of work, MacKenzie and Stout~\cite{mackenzie1998ultrafast} showed by recursively finding the closest pair on a sample of size $n^c$, the total work is $O(n)$ in expectation. We find $c=0.8$ to work well in practice.
In terms of depth, our implementation has $O(\log n)$ levels of recursion, each taking $O(\log^* n)$ depth \whp{}, which includes parallel dictionary operations and finding the minimum in parallel. The total depth is $O(\log n\log^*n)$ \whp{}.
In the recursion, we coarsen the base case by switching to a brute-force algorithm when the problem is sufficiently small.

\subsection{Sieve Algorithm}
Khuller and Matias~\cite{khuller1995simple} propose a simple sequential algorithm called the sieve algorithm that takes $O(n)$ expected work (the dynamic algorithm by Golin et al.~\cite{golin1998randomized} is based on the sieve algorithm).
The algorithm proceeds in rounds, where in round $i$, it chooses a random point $x$ from the point set $S_i$ (where $S_1=S$) and computes $d_i(x)$, the distance to its closest neighbor. Then, the algorithm constructs a grid structure on $S_i$, where each box has a side length of $d_i(x)$.
It then moves the points that are sparse in $S_i$ into a new set $S_{i+1}$, and proceeds to the next round, until $S_{i+1}$ is empty.
Finally, the algorithm constructs a grid structure on $S$ with boxes of size equal to the smallest box computed during the algorithm. For each point $x\in S$, we compute its closest neighbor using the grid by traversing its own box and the boxes bordering on it. Finally, we take the minimum among distances obtained by all the points to obtain $\delta(S)$.

The sequential algorithm takes $O(n)$ expected work as the number of points decreases geometrically from one level to the next.
We obtain a parallel sieve algorithm by using our parallel construction for the sparse partition in Algorithm~\ref{alg:build-short}, but without the heap.  Our parallel sieve algorithm takes $O(n)$ expected work and $O(\log n\log^* n)$ depth \whp.

\subsection{Incremental Algorithm}
Golin and Raman~\cite{golin1992simple} present a sequential incremental algorithm for closest pair with $O(n)$ expected work. 
Blelloch et al.~\cite{BGSS16} present a parallel version of this incremental algorithm, which we implement. 
The parallel algorithm works by maintaining a grid using a dictionary, and inserting the points in a randomized order in batches of exponentially increasing size. The side length of the grid box is the current closest pair distance, which is initialized to the distance between the first two points in the randomized ordering. For the $i$'th point inserted, the algorithm will check its neighborhood for a neighbor with distance smaller than the current grid side length. When such a neighbor is found, the algorithm rebuilds the grid for the first $i$ points using the new side length, and continues with the insertion. Since the parallel algorithm inserts points in batches, for each batch we find the earliest point $i$ remaining in the batch that causes a grid rebuild, perform the rebuild on all points up to and including $i$, remove these points from the batch, and repeat until the batch is empty.
After all batches are processed, the pair whose distance gives rise to the final grid side length is the closest pair. The algorithm takes $O(n)$ expected work and $O(\log n\log^* n)$ depth \whp.

\section{Experimental Data}\label{appendix:experiment}

\cref{table:throughput} shows the throughput of our parallel batch-dynamic algorithm on 1 thread and 48 cores with hyper-threading, as well as that of our sequential implementation, under varying batch sizes.

\begin{table}[]
\footnotesize
\setlength{\tabcolsep}{1.5pt}
\begin{tabular}{@{}c@{ }c|ccc|ccc|ccc@{}}
Batch Sizes                      &           & \multicolumn{3}{c|}{$10^2$}         & \multicolumn{3}{c|}{$10^3$}         & \multicolumn{3}{c}{$10^4$}           \\
                                 &           & Seq        & 1t         & 48h        & Seq        & 1t         & 48h        & Seq         & 1t          & 48h        \\ \hline
\multirow{2}{*}{2D-Uniform-10M}  & Ins & $2.75\cdot 10^5$ & $2.55\cdot 10^5$ & $1.40\cdot 10^5$ & $3.32\cdot 10^5$ & $3.18\cdot 10^5$ & $5.66\cdot 10^5$ & $2.44\cdot 10^5$  & $2.35\cdot 10^5$  & $1.66\cdot 10^6$ \\
                                 & Del  & $1.52\cdot 10^5$ & $1.52\cdot 10^5$ & $9.35\cdot 10^4$ & $1.56\cdot 10^5$ & $1.67\cdot 10^5$ & $3.09\cdot 10^5$ & $1.45\cdot 10^5$  & $1.55\cdot 10^5$  & $9.18\cdot 10^5$ \\
\multirow{2}{*}{3D-Uniform-10M}  & Ins & $9.00\cdot 10^4$ & $7.42\cdot 10^4$ & $6.13\cdot 10^4$ & $8.52\cdot 10^4$ & $7.24\cdot 10^4$ & $2.99\cdot 10^5$ & $7.53\cdot 10^4$  & $6.37\cdot 10^4$  & $9.52\cdot 10^5$ \\
                                 & Del  & $5.40\cdot 10^4$ & $4.77\cdot 10^4$ & $3.81\cdot 10^4$ & $5.54\cdot 10^4$ & $5.06\cdot 10^4$ & $1.55\cdot 10^5$ & $5.02\cdot 10^4$  & $4.59\cdot 10^4$  & $6.69\cdot 10^5$ \\
\multirow{2}{*}{5D-Uniform-10M}  & Ins & $4.83\cdot 10^3$ & $3.48\cdot 10^4$ & $4.08\cdot 10^4$ & $2.03\cdot 10^4$ & $2.69\cdot 10^4$ & $1.14\cdot 10^5$ & $5.00\cdot 10^4$  & $4.29\cdot 10^4$  & $5.60\cdot 10^5$ \\
                                 & Del  & $6.27\cdot 10^4$ & $6.16\cdot 10^4$ & $8.40\cdot 10^4$ & $2.98\cdot 10^4$ & $2.37\cdot 10^4$ & $2.07\cdot 10^5$ & $5.55\cdot 10^4$  & $4.46\cdot 10^4$  & $7.68\cdot 10^5$ \\
\multirow{2}{*}{7D-Uniform-10M}  & Ins & $4.38\cdot 10^3$ & $1.84\cdot 10^4$ & $3.67\cdot 10^4$ & $2.14\cdot 10^4$ & $3.01\cdot 10^4$ & $1.23\cdot 10^5$ & $4.37\cdot 10^4$  & $4.04\cdot 10^4$  & $4.65\cdot 10^5$ \\
                                 & Del  & $2.77\cdot 10^4$ & $2.48\cdot 10^4$ & $2.78\cdot 10^4$ & $4.78\cdot 10^4$ & $4.51\cdot 10^4$ & $3.73\cdot 10^5$ & $7.60\cdot 10^4$  & $7.26\cdot 10^4$  & $1.12\cdot 10^6$ \\
\multirow{2}{*}{2D-SS-varden-10M}   & Ins & $2.44\cdot 10^5$ & $2.25\cdot 10^5$ & $1.68\cdot 10^5$ & $3.99\cdot 10^5$ & $3.74\cdot 10^5$ & $6.17\cdot 10^5$ & $3.67\cdot 10^5$  & $3.53\cdot 10^5$  & $1.90\cdot 10^6$ \\
                                 & Del  & $9.88\cdot 10^4$ & $9.76\cdot 10^4$ & $1.02\cdot 10^5$ & $1.63\cdot 10^5$ & $1.69\cdot 10^5$ & $2.82\cdot 10^5$ & $1.83\cdot 10^5$  & $1.92\cdot 10^5$  & $1.21\cdot 10^6$ \\
\multirow{2}{*}{3D-SS-varden-10M}   & Ins & $1.33\cdot 10^5$ & $1.12\cdot 10^5$ & $9.55\cdot 10^4$ & $1.11\cdot 10^5$ & $9.56\cdot 10^4$ & $3.74\cdot 10^5$ & $1.05\cdot 10^5$  & $8.89\cdot 10^4$  & $1.14\cdot 10^6$ \\
                                 & Del  & $1.00\cdot 10^5$ & $9.20\cdot 10^4$ & $6.32\cdot 10^4$ & $8.40\cdot 10^4$ & $8.45\cdot 10^4$ & $4.10\cdot 10^5$ & $8.33\cdot 10^4$  & $8.20\cdot 10^4$  & $1.04\cdot 10^6$ \\
\multirow{2}{*}{5D-SS-varden-10M}   & Ins & $1.16\cdot 10^5$ & $7.51\cdot 10^4$ & $4.85\cdot 10^4$ & $1.25\cdot 10^5$ & $1.04\cdot 10^5$ & $2.42\cdot 10^5$ & $1.20\cdot 10^5$  & $1.03\cdot 10^5$  & $7.19\cdot 10^5$ \\
                                 & Del  & $1.84\cdot 10^5$ & $1.74\cdot 10^5$ & $1.05\cdot 10^5$ & $2.04\cdot 10^5$ & $1.94\cdot 10^5$ & $3.92\cdot 10^5$ & $2.05\cdot 10^5$  & $1.94\cdot 10^5$  & $1.15\cdot 10^6$ \\
\multirow{2}{*}{7D-SS-varden-10M}   & Ins & $1.01\cdot 10^5$ & $6.78\cdot 10^4$ & $4.61\cdot 10^4$ & $1.11\cdot 10^5$ & $8.63\cdot 10^4$ & $2.43\cdot 10^5$ & $1.09\cdot 10^5$  & $8.58\cdot 10^4$  & $5.87\cdot 10^5$ \\
                                 & Del  & $1.86\cdot 10^5$ & $1.77\cdot 10^5$ & $1.14\cdot 10^5$ & $2.12\cdot 10^5$ & $1.83\cdot 10^5$ & $4.36\cdot 10^5$ & $2.14\cdot 10^5$  & $1.82\cdot 10^5$  & $1.30\cdot 10^6$ \\
\multirow{2}{*}{7D-Household-2M} & Ins & --         & --         & $8.15\cdot 10^2$ & $7.94\cdot 10^2$ & $5.18\cdot 10^2$ & $8.94\cdot 10^3$ & $2.44\cdot 10^3$  & $1.34\cdot 10^3$  & $2.66\cdot 10^4$ \\
                                 & Del  & --         & --         & $3.66\cdot 10^3$ & $2.80\cdot 10^4$ & $1.36\cdot 10^4$ & $3.16\cdot 10^4$ & $1.42\cdot 10^3$  & $1.88\cdot 10^3$  & $4.19\cdot 10^4$ \\
\multirow{2}{*}{16D-Chem-4M}     & Ins & --         & --         & $4.20\cdot 10^2$ & $2.62\cdot 10^4$ & $2.86\cdot 10^4$ & $6.12\cdot 10^4$ & $4.88\cdot 10^4$  & $6.20\cdot 10^4$  & $5.74\cdot 10^5$ \\
                                 & Del  & --         & --         & $6.95\cdot 10^2$ & $7.00\cdot 10^4$ & $7.93\cdot 10^4$ & $3.15\cdot 10^5$ & $1.48\cdot 10^5$  & $1.69\cdot 10^5$  & $1.23\cdot 10^6$ \\
\multirow{2}{*}{3D-Cosmo-298M}   & Ins & --         & --         & $4.96\cdot 10^4$ & --         & --         & $6.36\cdot 10^4$ & --          & --          & $2.73\cdot 10^5$ \\
                                 & Del  & --         & --         & $2.62\cdot 10^4$ & --         & --         & $1.86\cdot 10^5$ & --          & --          & $2.94\cdot 10^5$ \\ \hline
Batch Sizes                      &           & \multicolumn{3}{c|}{$10^5$}         & \multicolumn{3}{c|}{$10^6$}         & \multicolumn{3}{c}{$\min\{n,10^7\}$} \\
                                 &           & Seq        & 1t         & 48h        & Seq        & 1t         & 48h        & Seq         & 1t          & 48h        \\ \hline
\multirow{2}{*}{2D-Uniform-10M}  & Ins & $2.30\cdot 10^5$ & $2.20\cdot 10^5$ & $3.55\cdot 10^6$ & $3.47\cdot 10^5$ & $3.34\cdot 10^5$ & $9.67\cdot 10^6$ & $4.13\cdot 10^5$  & $3.92\cdot 10^5$  & $1.30\cdot 10^7$ \\
                                 & Del  & $1.09\cdot 10^5$ & $1.14\cdot 10^5$ & $1.54\cdot 10^6$ & $1.48\cdot 10^5$ & $1.57\cdot 10^5$ & $5.24\cdot 10^6$ & $2.42\cdot 10^5$  & $2.55\cdot 10^5$  & $9.29\cdot 10^6$ \\
\multirow{2}{*}{3D-Uniform-10M}  & Ins & $8.09\cdot 10^4$ & $6.15\cdot 10^4$ & $1.70\cdot 10^6$ & $9.97\cdot 10^4$ & $7.53\cdot 10^4$ & $2.37\cdot 10^6$ & $1.04\cdot 10^5$  & $1.02\cdot 10^5$  & $3.88\cdot 10^6$ \\
                                 & Del  & $4.68\cdot 10^4$ & $4.19\cdot 10^4$ & $1.20\cdot 10^6$ & $6.60\cdot 10^4$ & $6.15\cdot 10^4$ & $2.04\cdot 10^6$ & $9.59\cdot 10^4$  & $9.56\cdot 10^4$  & $3.69\cdot 10^6$ \\
\multirow{2}{*}{5D-Uniform-10M}  & Ins & $7.21\cdot 10^4$ & $5.96\cdot 10^4$ & $1.62\cdot 10^6$ & $1.03\cdot 10^5$ & $8.85\cdot 10^4$ & $2.20\cdot 10^6$ & $9.82\cdot 10^4$  & $1.23\cdot 10^5$  & $2.04\cdot 10^6$ \\
                                 & Del  & $7.79\cdot 10^4$ & $6.34\cdot 10^4$ & $1.86\cdot 10^6$ & $1.20\cdot 10^5$ & $1.04\cdot 10^5$ & $2.99\cdot 10^6$ & $1.65\cdot 10^5$  & $1.82\cdot 10^5$  & $5.09\cdot 10^6$ \\
\multirow{2}{*}{7D-Uniform-10M}  & Ins & $6.71\cdot 10^4$ & $6.54\cdot 10^4$ & $1.56\cdot 10^6$ & $1.20\cdot 10^5$ & $1.42\cdot 10^5$ & $3.92\cdot 10^6$ & $1.00\cdot 10^5$  & $1.34\cdot 10^5$  & $4.21\cdot 10^6$ \\
                                 & Del  & $1.48\cdot 10^5$ & $1.52\cdot 10^5$ & $3.80\cdot 10^6$ & $2.91\cdot 10^5$ & $2.81\cdot 10^5$ & $6.84\cdot 10^6$ & $2.69\cdot 10^5$  & $3.11\cdot 10^5$  & $1.06\cdot 10^7$ \\
\multirow{2}{*}{2D-SS-varden-10M}   & Ins & $3.82\cdot 10^5$ & $3.83\cdot 10^5$ & $6.61\cdot 10^6$ & $4.39\cdot 10^5$ & $4.17\cdot 10^5$ & $1.21\cdot 10^7$ & $3.39\cdot 10^5$  & $3.70\cdot 10^5$  & $1.35\cdot 10^7$ \\
                                 & Del  & $1.66\cdot 10^5$ & $1.75\cdot 10^5$ & $3.55\cdot 10^6$ & $1.72\cdot 10^5$ & $1.81\cdot 10^5$ & $5.42\cdot 10^6$ & $2.43\cdot 10^5$  & $2.49\cdot 10^5$  & $9.36\cdot 10^6$ \\
\multirow{2}{*}{3D-SS-varden-10M}   & Ins & $9.66\cdot 10^4$ & $7.55\cdot 10^4$ & $2.25\cdot 10^6$ & $1.03\cdot 10^5$ & $7.69\cdot 10^4$ & $2.18\cdot 10^6$ & $1.00\cdot 10^5$  & $1.03\cdot 10^5$  & $4.07\cdot 10^6$ \\
                                 & Del  & $7.61\cdot 10^4$ & $6.93\cdot 10^4$ & $2.01\cdot 10^6$ & $7.56\cdot 10^4$ & $7.03\cdot 10^4$ & $2.15\cdot 10^6$ & $9.38\cdot 10^4$  & $9.80\cdot 10^4$  & $3.75\cdot 10^6$ \\
\multirow{2}{*}{5D-SS-varden-10M}   & Ins & $1.11\cdot 10^5$ & $1.00\cdot 10^5$ & $1.57\cdot 10^6$ & $9.12\cdot 10^4$ & $9.97\cdot 10^4$ & $2.10\cdot 10^6$ & $1.04\cdot 10^5$  & $1.32\cdot 10^5$  & $1.64\cdot 10^6$ \\
                                 & Del  & $1.91\cdot 10^5$ & $1.95\cdot 10^5$ & $3.21\cdot 10^6$ & $1.64\cdot 10^5$ & $1.89\cdot 10^5$ & $4.34\cdot 10^6$ & $1.58\cdot 10^5$  & $1.81\cdot 10^5$  & $3.95\cdot 10^6$ \\
\multirow{2}{*}{7D-SS-varden-10M}   & Ins & $9.40\cdot 10^4$ & $8.28\cdot 10^4$ & $1.38\cdot 10^6$ & $8.50\cdot 10^4$ & $8.04\cdot 10^4$ & $1.44\cdot 10^6$ & $8.70\cdot 10^4$  & $1.03\cdot 10^5$  & $1.91\cdot 10^6$ \\
                                 & Del  & $1.89\cdot 10^5$ & $1.76\cdot 10^5$ & $3.74\cdot 10^6$ & $1.62\cdot 10^5$ & $1.68\cdot 10^5$ & $4.32\cdot 10^6$ & $1.44\cdot 10^5$  & $1.59\cdot 10^5$  & $4.24\cdot 10^6$ \\
\multirow{2}{*}{7D-Household-2M} & Ins & $6.31\cdot 10^3$ & $5.72\cdot 10^3$ & $9.04\cdot 10^4$ & $5.98\cdot 10^3$ & $9.25\cdot 10^3$ & $1.80\cdot 10^5$ & $2.09\cdot 10^4$  & $3.72\cdot 10^4$  & $5.43\cdot 10^5$ \\
                                 & Del  & $4.43\cdot 10^3$ & $1.13\cdot 10^4$ & $2.08\cdot 10^5$ & $9.21\cdot 10^3$ & $3.92\cdot 10^4$ & $6.35\cdot 10^5$ & $5.63\cdot 10^4$  & $9.32\cdot 10^4$  & $1.92\cdot 10^6$ \\
\multirow{2}{*}{16D-Chem-4M}     & Ins & $4.51\cdot 10^4$ & $5.95\cdot 10^4$ & $9.19\cdot 10^5$ & $4.34\cdot 10^4$ & $6.49\cdot 10^4$ & $8.73\cdot 10^5$ & $4.10\cdot 10^4$  & $6.38\cdot 10^4$  & $7.89\cdot 10^5$ \\
                                 & Del  & $1.35\cdot 10^5$ & $1.61\cdot 10^5$ & $1.27\cdot 10^6$ & $1.30\cdot 10^5$ & $1.76\cdot 10^5$ & $2.78\cdot 10^6$ & $1.21\cdot 10^5$  & $1.75\cdot 10^5$  & $4.52\cdot 10^6$ \\
\multirow{2}{*}{3D-Cosmo-298M}   & Ins & --         & --         & $2.40\cdot 10^5$ & --         & --         & $6.20\cdot 10^5$ & --          & --          & $8.95\cdot 10^5$ \\
                                 & Del  & --         & --         & $3.19\cdot 10^5$ & --         & --         & $6.61\cdot 10^5$ & --          & --          & $9.18\cdot 10^5$
\end{tabular}
\caption{Throughput (number of points processed per second) of the batch-dynamic algorithm with varying batch sizes. "Seq" denotes the sequential implementation, "1t" denotes the parallel implementation run on 1 thread, and "48h" denotes the parallel implementation run on 48 cores with two-way hyper-threading. "Ins" and "Del" denote the throughput for batch insertion and batch deletion, respectively.
}\label{table:throughput}
\end{table}

\section{Related Work}\label{sec:related}

\myparagraph{Static Closest Pair}
The problem of finding the closest pair of points under the $L_t$ metric has been a long-studied problem in computational geometry.
There have been several deterministic sequential algorithms~\cite{shamos1975closest, bentley1976divide, bentley1980multidimensional, hinrichs1988plane} that solve the problem optimally in $O(n\log n)$ time under the standard algebraic decision tree model.
Under a different model where the floor function is allowed with unit-cost, Rabin~\cite{rabin1976probabilistic} solved the problem in $O(n)$ expected time. 
Fortune and Hopcroft~\cite{fortune1979note} presented a deterministic algorithm with $O(n\log \log n)$ running time under the same model. 
Later, Khuller and Matias~\cite{khuller1995simple} came up with a simple randomized algorithm that takes $O(n)$ expected time using a sieve data structure. 
Golin et al.~\cite{golin1992simple} described a randomized incremental algorithm for the problem that takes $O(n)$ expected time. 

Later, Dietzfelbinger et al.~\cite{dietzfelbinger1997reliable} filled in some details for Rabin's algorithm concerning hashing and duplicate grouping. Banyassady and Mulzer~\cite{banyassady2017simple} gave a simpler analysis for Rabin's algorithm. 
Chan~\cite{chan1999geometric} gave an algorithm that takes $O(n)$ expected time in a randomized optimization framework. Maheshwari et al.~\cite{maheshwari2020simple} designed a new algorithm for closest pair in the doubling-metric space in expected $O(n\log n)$ time.

For parallel algorithms, Atallah and Goodrich~\cite{atallah1986efficient} came up with the first parallel algorithm for closest pair using multi-way divide-and-conquer. The algorithm takes $O(n\log n\log\log n)$ work and $O(\log n \log\log n)$ depth. 
MacKenzie and Stout~\cite{mackenzie1998ultrafast} designed a parallel algorithm inspired by Rabin~\cite{rabin1976probabilistic} that takes $O(n)$ work and $O(1)$ depth in expectation. 
Blelloch and Maggs~\cite{blelloch2010parallel} parallelized the divide-and-conquer approach in~\cite{bentley1976divide,bentley1980multidimensional}, and their algorithm takes $O(n\log n)$ work and $O(\log^2 n)$ depth. 
Recently, Blelloch et al.~\cite{BGSS16} designed a randomized incremental algorithm for the problem that takes $O(n)$ expected work and $O(\log n\log^* n)$ depth \whp. 
Lenhof and Smid~\cite{lenhof1995sequential} solved a close variant, the $K$-closest pair problem for the $L_t$-metric, in $O(n\log n \log\log n +K)$ work and $O(\log^2 n \log\log n)$ depth, where $K$ is the number of closest pairs to return. 

\myparagraph{Dynamic Closest Pair} 
For the sequential dynamic closest pair problem, there have been semi-dynamic algorithms that focus on only insertions or only deletions.
For the deletion-only case,
Supowit~\cite{supowit1990new} gave an algorithm that maintains the minimal Euclidean distance for points in $k$ dimensions in $O(\log^k n)$ amortized time per deletion. The method uses $O(n\log^{k-1} n)$ space. 
For the insertion-only case, Schwarz and Smid~\cite{schwarz1991n} designed a data structure for the $L_t$-metric space that takes $O(n)$ space and handles each insertion in $O(\log n\log\log n)$ time. 
Schwarz et al.~\cite{schwarz1994optimal} designed an optimal data structure taking $O(n)$ space and $O(\log n)$ time per insertion for the same problem. 

For sequential fully-dynamic closest pair algorithms supporting both insertions and deletions,
Overmars~\cite{overmars1981dynamization} gave an $O(n)$ time update algorithm.
Smid~\cite{smid1990maintaining} gave a dynamic data structure of size $O(n)$, that maintains closest pair of $k$-dimensional points in the $L_t$-metric space, in $O(n^{2/3} \log n)$ time per update. 

Later work improved the sequential running time to polylogarithmic time per update.
Smid~\cite{smid1991maintaining} used a data structure of size $O(n\log^k n)$ that maintains the closest pair in the $L_t$-metric in $O(\log^k n \log \log n)$ amortized time per update. 
Callahan and Kosaraju~\cite{callahan1995algorithms} presented a general technique for dynamizing problems in Euclidean-space that make use of the well-separated pair decomposition~\cite{callahan1995decomposition}. 
For dynamic closest pair, their algorithm requires $O(n)$ space and $O(\log^2 n)$ time per update.
Kapoor and Smid~\cite{kapoor1996new} proposed new dynamic data structures that solve the problem in the $L_t$-metric using $O(n)$ space and deterministic $O(\log^{k-1}n \log\log n)$ update time. 
Bespamyatnikh~\cite{bespamyatnikh1998optimal} described a closest pair data structure for the $L_t$-metric space that takes $O(n)$ space, and has $O(\log n)$ deterministic update time. The main idea is to dynamically maintain a fair-split tree and a heap of neighbor pairs. The algorithm does not currently seem to be practical. 
Golin et al.~\cite{golin1998randomized} described a randomized data structure for the problem in the $L_t$-metric space, which takes $O(n)$ expected space and supports  insertions and deletions in $O(\log n)$ expected time. 
Our work is inspired by Golin et al.~\cite{golin1998randomized}, and is the first to parallelize dynamic batch-updates for the closest pair problem.

In addition, there has been related work on similar problems. Agarwal et al.~\cite{agarwal2008kinetic} proposed a kinetic and dynamic data structure for 2-dimensional
all-nearest-neighbors as well as the closest pair, which uses $O(n\log n)$ space and processes each update in $O(\log^3 n)$ time. 
Eppstein~\cite{eppstein2000fast} solved a stronger version of the dynamic closest pair problem by supporting arbitrary distance functions. The algorithm maintains the closest pair in $O(n \log n)$ time per insertion and $O(n \log^2 n)$ amortized time per deletion using $O(n)$ space.
Cardinal and Eppstein~\cite{cardinal2004lazy} later designed a more practical version of this algorithm.
Chan~\cite{chan2020dynamic} presented a modification of Eppstein's algorithm~\cite{eppstein2000fast}, which improves the amortized update time to $O(n)$.

}{}

\end{document}